\documentclass[11pt]{article}
\usepackage[utf8]{inputenc}
\usepackage[margin=0.8in]{geometry}

\usepackage{svg}
\usepackage{xcolor}
\usepackage[]{algorithm}
\usepackage[noend]{algpseudocode}
\usepackage{graphicx,nicefrac}
\usepackage{amsmath,amsthm,amssymb,caption,nicefrac}
\usepackage{subcaption}
\usepackage{dsfont}
\usepackage{bbm}
\usepackage[font={small}]{caption}
\usepackage{thm-restate}
\usepackage[normalem]{ulem}
\usepackage{float}
\usepackage{balance}
\usepackage{url}
\usepackage{pifont}
\usepackage{wrapfig}

\usepackage{array}

\usepackage{enumitem}

\usepackage{enumitem}
\usepackage{xcolor}

\usepackage[hidelinks]{hyperref}
\usepackage{cleveref}

\newcommand{\nc}{\newcommand}
\newcommand{\DMO}{\DeclareMathOperator}
\DeclareMathAlphabet\mathbfcal{OMS}{cmsy}{b}{n}

\newcommand{\Description}{\caption}

\algdef{SE}[REPEATN]{RepeatN}{End}[1]{\algorithmicrepeat\ #1 \textbf{times}}{\algorithmicend}

\usepackage{xcolor}
\definecolor{Gred}{RGB}{219, 50, 54}
\definecolor{Ggreen}{RGB}{60, 186, 84}
\definecolor{Gblue}{RGB}{72, 133, 237}
\definecolor{Gyellow}{RGB}{247, 178, 16}
\definecolor{ToCgreen}{RGB}{0, 128, 0}
\definecolor{myGold}{RGB}{231,141,20}
\definecolor{myBlue}{rgb}{0.19,0.41,.65}
\definecolor{myPurple}{RGB}{175,0,124}

\providecommand{\Comments}{0}
\ifnum\Comments>0
\usepackage[colorinlistoftodos,prependcaption,textsize=scriptsize]{todonotes}
\paperwidth=\dimexpr \paperwidth + 10cm\relax
\oddsidemargin=\dimexpr\oddsidemargin + 5cm\relax
\evensidemargin=\dimexpr\evensidemargin + 5cm\relax
\newcommand{\modified}[1]{\textcolor{blue}{#1}}

\else
\usepackage[disable]{todonotes}
\newcommand{\modified}[1]{{#1}}

\fi

\newcommand{\mytodo}[1]{\ifnum\Comments=1{#1}\fi}

\newcommand{\tableoftodos}{\ifnum\Comments=1 \listoftodos[Comments/To Do's] \fi}

\allowdisplaybreaks

\nc{\MS}{\mathcal{S}}
\nc{\MP}{\mathcal{P}}
\nc{\cM}{\mathcal{M}}
\nc{\cS}{\mathcal{S}}
\nc{\cI}{\mathcal{I}}
\nc{\cA}{\mathcal{A}}
\nc{\tcA}{\tilde{\cA}}

\nc{\MZ}{\mathcal{Z}}
\newcommand{\cX}{\mathcal{X}}
\DMO{\Binom}{Binom}
\newcommand{\E}{\mathbb{E}}
\DMO{\Var}{Var}

\newcommand{\bX}{\mathbf{X}}
\nc{\tbx}{\tilde{\bx}}
\nc{\tbX}{\tilde{\bX}}
\nc{\tZ}{\tilde{Z}}
\nc{\tz}{\tilde{z}}
\newcommand{\bU}{\mathbf{U}}
\nc{\tbU}{\tilde{\bU}}
\newcommand{\bT}{\mathbf{T}}
\nc{\tbT}{\tilde{\bT}}
\newcommand{\bD}{\mathbf{D}}
\nc{\tbD}{\tilde{\bD}}

\newcommand{\bx}{\mathbf{x}}

\newcommand{\R}{\mathbb{R}}
\newcommand{\N}{\mathbb{N}}

\nc{\BN}{\mathbb{N}}
\newcommand{\Z}{\mathbb{Z}}
\nc{\BZ}{\mathbb{Z}}
\newcommand{\cR}{\mathcal{R}}

\newcommand{\disc}{\mathrm{d}}
\newcommand{\bA}{\mathbf{A}}
\nc{\tbA}{\tilde{\bA}}

\newcommand{\cD}{\mathcal{D}}

\newcommand{\tu}{\tilde{u}}

\newcommand{\cP}{\mathcal{P}}

\DeclareMathOperator{\supp}{supp}

\DeclareMathOperator{\DLap}{DLap}
\DeclareMathOperator{\Lap}{Lap}
\DeclareMathOperator{\DirM}{DirM}
\DeclareMathOperator{\TV}{D_{\mathrm{TV}}}
\DeclareMathOperator{\Ber}{Ber}

\DeclareMathOperator{\DStair}{DStair}

\DeclareMathOperator{\NB}{NB}

\DeclareMathOperator{\GDL}{GDL}

\DeclareMathOperator{\Geo}{Geo}
\DeclareMathOperator{\MSDLap}{MSDLap}
\DeclareMathOperator{\Arete}{Arete}

\newcommand{\eps}{\varepsilon}

\newcommand{\pr}[1]{\mathbb{P}\left[#1\right]}

\newcommand{\cmark}{\text{\ding{51}}}%
\newcommand{\xmark}{\text{\ding{55}}}%

\newcommand{\HG}[4]{{}_2F_1\left[#1, #2; #3; #4\right]}
\newcommand{\dr}[2]{\mathrm{D}_{\infty}\left(#1\ \middle\|\ #2\right)}

\newtheorem{thm}{Theorem}
\newtheorem{theorem}[thm]{Theorem}
\newtheorem{observation}[thm]{Observation}
\newtheorem{lemma}[thm]{Lemma}

\newtheorem{definition}[thm]{Definition}

\newtheorem{corollary}[thm]{Corollary}
\newtheorem{remark}[thm]{Remark}
\newtheorem{conjecture}[thm]{Conjecture}
\newtheorem{proposition}[thm]{Proposition}

\crefname{observation}{Observation}{Observations}

\DeclareRobustCommand
  \Compactcdots{\mathinner{\cdotp\mkern-2mu\cdotp\mkern-2mu\cdotp}}

\title{Infinitely Divisible Noise for Differential Privacy:\\ Nearly Optimal Error in the High $\eps$ Regime}

\author{Charlie Harrison \\Google\\ \texttt{\small csharrison@google.com} \and Pasin Manurangsi \\Google Research\\ \texttt{\small pasin@google.com}}

\begin{document}

\maketitle

\begin{abstract}
Differential privacy can be achieved in a distributed manner, where multiple parties add independent noise such that their sum protects the overall dataset with differential privacy. A common technique here is for each party to sample their noise from the decomposition of an infinitely divisible distribution.
We analyze two mechanisms in this setting:
1) the generalized discrete Laplace (GDL) mechanism, whose distribution (which is closed under summation) follows from differences of i.i.d. negative binomial shares, and 2) the multi-scale discrete Laplace (MSDLap) mechanism, a novel mechanism following the sum of multiple i.i.d. discrete Laplace shares at different scales.
For $\eps \geq 1$, our mechanisms can be parameterized to have $O\left(\Delta^3 e^{-\eps}\right)$ and $O\left(\min\left(\Delta^3 e^{-\eps}, \Delta^2 e^{-2\eps/3}\right)\right)$ MSE, respectively, where $\Delta$ denote the sensitivity; the latter bound matches known optimality results. Furthermore, the MSDLap mechanism has the optimal MSE \emph{including constants} as $\eps \to \infty$.
We also show a transformation from the discrete setting to the continuous setting, which allows us to transform both mechanisms to the continuous setting and thereby achieve the optimal $O\left(\Delta^2 e^{-2\eps / 3}\right)$ MSE.
To our knowledge, these are the first infinitely divisible additive noise mechanisms that achieve order-optimal MSE under pure differential privacy for either the discrete or continuous setting, so our work shows formally there is no separation in utility when query-independent noise adding mechanisms are restricted to infinitely divisible noise. 
For the continuous setting, our result improves upon Pagh and Stausholm's Arete distribution which gives an MSE of $O\left(\Delta^2 e^{-\eps/4}\right)$ ~\cite{pagh2022infinitely}.
Furthermore, we give an exact sampler tuned to efficiently implement  the MSDLap mechanism, and we apply our results to improve a state of the art multi-message shuffle DP protocol from \cite{balle2020private} in the high $\eps$ regime.
\end{abstract}

%\keywords{Differential privacy, multi-party computation, secure aggregation}

\maketitle

\section{Introduction}

Differential Privacy (DP) \cite{DworkMNS06} is a %widely accepted, 
formal notion of privacy
which bounds the sensitive information revealed by an algorithm. While there are many
"flavors" of differential privacy, most relevant to this work is so-called pure-DP or $\eps$-DP which bounds the privacy loss via the parameter $\eps > 0$.

\begin{definition}%[Differential privacy 
[\cite{DworkMNS06}]
A randomized mechanism $M : \mathcal{X}^d \to \mathcal{Y}$ satisfies $\eps$-differential privacy if, for all $x,x' \in \mathcal{X}^d$ differing\footnote{For the purpose of this work, we may consider any neighboring notion. We only use the substitution notion for simplicity.} in a single entry,
$
\Pr[M(x) \in S] \leq e^\eps \cdot \Pr[M(x') \in S]
$
for all measurable $S \subseteq \mathcal{Y}$.
\end{definition}
%In this work, 
We focus on the so-called \emph{low-privacy regime} where $\eps \geq 1$. Despite its name, this regime still provides meaningful privacy protection and is the setting most often employed in practical applications of DP (e.g.~\cite{abowd2018us,dp2017learning,gboard}). The utility bounds we state throughout will focus on this regime.

A challenge in deploying differential privacy is doing so while also producing \emph{useful} results. In this work, we focus on minimizing the mean squared error (MSE) of a query $q$ subject to a query-independent additive noise mechanism, i.e. $M(x) = q(x) + Z$ where $Z$ is a random variable. There is a rich body of work on optimizing MSE (and in particular the MSE's dependence on $\eps$) in this setting.
Notably, the staircase mechanism (\cite{geng2014optimal}, \cite{stair2015}, and \cite{stair2016adding}) was shown to have the optimal MSE of all differentially private, query-independent additive noise mechanisms. In the continuous setting, it achieves a MSE of $O\left(\Delta^2 e^{-2\eps / 3}\right)$, and in the discrete setting it can achieve a MSE interpolating between $O\left(\Delta^3 e^{-\eps}\right)$ and $O\left(\Delta^2 e^{-2\eps/3}\right)$.
For $\Delta = 1$, the optimal discrete staircase mechanism is the discrete Laplace mechanism (also known as the Geometric mechanism)~\cite{GhoshRS12}.

A probability distribution $\cD$ is \emph{infinitely divisible} iff for every positive integer $n$, there exists a distribution $\cD_{/n}$ such that, when we sample $n$ i.i.d. random variables $Z_1, \cdots, Z_n \sim \cD_{/n}$,  their sum $Z = \sum_{i=1}^n Z_i$ is distributed as $\cD$. In distributed differential privacy, a common technique (see \cite{goryczka2015comprehensive} for an overview) is for $n$ parties to each sample $Z_i$ such that the sum is distributed according to $\cD$, which can be shown to protect the dataset with differential privacy. The infinite divisibility property of $\cD$ allows for distributed protocols where an arbitrary $n \ge 1$ number of parties can participate. Under the more restrictive setting where the additive noise mechanism $M$ must sample the noise $Z$ from an \emph{infinitely divisible} distribution, there was previously no known mechanism in either the discrete or continuous settings which matched the MSE of the staircase mechanism. We resolve this gap in this paper for both settings.

\subsection{Related work}

Distributed noise generation for differential privacy is well studied even for distributions that are not infinitely divisible. In fact, the idea dates back to %the work of Dwork et al.~\cite{DworkKMMN06} from 
the very early days of DP~\cite{DworkKMMN06}. Moreover, several works have studied the setting where $Z_1, \dots, Z_n$ samples from some distribution $\tilde{\cD}$ and directly argue about the distribution of their summation $Z = Z_1 + \cdots + Z_n$. Examples include the case where $\tilde{\cD}$ is a Bernoulli distribution~\cite{CheuSUZZ19,GhaziGKPV21}, for which $Z$ is a Binomial random variable, and the case where $\tilde{\cD}$ is a discrete Gaussian distribution~\cite{kairouz2021distributed}, for which $Z$ is ``close'' to discrete Gaussian random variable. The drawback here is that the distribution of $Z$ are different for different values of $n$, meaning that the privacy analysis often requires $n$ to be sufficiently large (e.g.~\cite{GhaziGKPV21}) or sufficiently small (e.g.~\cite{kairouz2021distributed}). Using infinitely divisible distribution overcomes this issue since the distribution of the total noise $Z$ is always $\cD$ regardless of the value of $n$, leading to a privacy analysis that works for all regimes of $n$. Due to this, infinitely divisible noise distributions have gained popularity in distributed settings of differential privacy (e.g.~\cite{goryczka2015comprehensive,balle2020private,GhaziKMP20,almost-central-shuffle-ghazi21a,agarwal2021skellam,CanonneL22,bagdasaryan2022heatmaps,GhaziKM24}).

As discrete distributions are typically easier to embed in multi-party cryptographic
protocols and avoid implementation issues \cite{mironov2012significance} with floating point representations, they tend to be more well-studied in distributed differential privacy.
The infinite divisibility of the discrete Laplace into negative binomial\footnote{Throughout this work, we use the term negative binomial to refer to the distribution generalized to a real valued stopping parameter $r$. In other works, this is sometimes called the P\'{o}lya distribution.} shares has been studied extensively in \cite{goryczka2015comprehensive, balle2020private, bagdasaryan2022heatmaps}. In \cite{bagdasaryan2022heatmaps} the authors explicitly analyzed privacy in the face of \emph{dropouts}, or parties that fail to properly add their noise share. To account for this, the noise shares from each individual party were scaled up such that the final noise distribution
is \emph{at least} discrete Laplace as long as no more than a constant fraction of parties drop out. This allowed the proposed system to continue to use the base analysis of the discrete Laplace mechanism as-is, by paying the price of increased noise.

In the continuous setting, the infinitely divisible Arete distribution was introduced in \cite{pagh2022infinitely} specifically to target the low-privacy, high $\eps$ pure-DP regime. It was designed to match the performance of the (continuous) staircase mechanism which is not infinitely divisible and therefore unusable in
the distributed setting. While the continuous staircase mechanism
achieves $O(\Delta^2 e^{-2\eps / 3})$ MSE, the authors only proved the Arete mechanism has an MSE of $O(\Delta^2 e^{-\eps / 4})$, though we believe this is not tight (\Cref{conj:arete}).

Distributed noise is relevant in other notions of differential privacy as well. The aforementioned discrete Gaussian mechanism \cite{canonne2020discrete, kairouz2021distributed} satisfies zero-concentrated DP \cite{bun2016concentrated}, and the Skellam mechanism \cite{agarwal2021skellam, valovich2017computational, schein2019locally} satisfies R\'{e}nyi DP \cite{mironov2017renyi}.\footnote{Mechanisms satisfying R\'{e}nyi DP or zero-concentrated DP naturally also satisfy approximate DP.}
Both mechanisms were proposed in the context of federated learning via a secure aggregation multi-party protocol \cite{kairouz2021practicalDPFTRL, secagg, secaggpolylog}.

In the context of pure differential privacy, the GDL mechanism was explored informally in a prior blog post by one of the authors~\cite{harrisongdl}. Concurrent to this work, the GDL distribution was studied in the context of $(\eps, \delta)$ shuffle differential privacy in \cite{athanasiou2025metricprivacy}. There the authors analyzed a \emph{shifted and truncated} version of the GDL distribution that achieves an approximate DP bound guarantee, vs. the pure DP one in this work which does not perform any truncation. Their primary result in the shuffle setting involving nearly matching the utility of the central discrete Laplace mechanism. On the other hand, as explained below, we \emph{improve on} the central discrete Laplace mechanism for sufficiently large $\eps$.

\subsection{Our contributions}

\bgroup
\def\arraystretch{1.3}%  1 is the default
\begin{table*}[h]
\centering
Discrete Distributions \\
\vspace{0.1cm}
\begin{tabular}{| m{0.34\linewidth} | m{0.23\linewidth} | m{0.09\linewidth} | m{0.24\linewidth}|}
\hline
Distribution & MSE & Inf. Div. & Reference \\
\hline
Discrete Laplace & $e^{-\eps / \Delta}$ & $\cmark$ & \cite{GhoshRS12} \\ 
\hline
 Discrete Staircase & $\min\{\Delta^3 e^{-\eps}, \Delta^2 e^{-2\eps / 3}\}$ & $\xmark$ & \cite{geng2014optimal, stair2015, stair2016adding}\\
 \hline
  \textbf{Generalized Discrete Laplace (GDL)} & $\Delta^3 e^{-\eps}$ for $\eps > 2 + \log(\Delta)$ & $\cmark$  & \Cref{thm:opt-var}\\
\hline
 \textbf{Multi-Scale Discrete Laplace (MSDLap)} & $\min\{\Delta^3 e^{-\eps}, \Delta^2 e^{-2\eps / 3}\}$ & $\cmark$ & \Cref{cor:msdlap-rextremes}\\
 \hline
 \end{tabular}

\vspace{0.2cm}
Continuous Distributions \\
\vspace{0.1cm}
 
\begin{tabular}{| m{0.34\linewidth} | m{0.23\linewidth} | m{0.09\linewidth} | m{0.24\linewidth}|}
\hline
Distribution & MSE & Inf. Div. & Reference \\
\hline
 Laplace & $\Delta^2 / \eps^2$ & $\cmark$ & \cite{DworkMNS06}\\  
 \hline
Staircase & $\Delta^2 e^{-2\eps / 3}$ & $\xmark$ & \cite{geng2014optimal, stair2015, stair2016adding}\\
 \hline
Arete & $\Delta^2 e^{-\eps / 4}$ & $\cmark$ & \cite{pagh2022infinitely}\\
\hline
 \textbf{Continuous MSDLap} & $\Delta^2 e^{-2\eps / 3}$ & $\cmark$ & \Cref{thm:discrete-to-cont}\\
\hline
\end{tabular}
\caption{Summary of our results (bold) compared to known noise distributions satisfying $\eps$-DP. Distributions are grouped by whether they are discrete (support on $\mathbb{Z}$) or continuous (support on $\mathbb{R}$). MSEs exclude constant factors.} \label{tab:summary}
\end{table*}
\egroup

In \Cref{sec:discrete}, we introduce the GDL and MSDLap mechanisms, two discrete noise-adding mechanisms having optimal $O(\Delta^3 e^{-\eps})$ MSE for fixed $\Delta$ and any sufficiently large $\eps$. Inspired by the discrete staircase mechanism, we also extend the MSDLap mechanism with a parameter optionally allowing it to satisfy $O(\Delta^2 e^{-2\eps/3})$ MSE, allowing it to achieve asymptotically-optimal error for any fixed $\Delta$ and $\eps \geq 1$.
Notably, the MSDLap mechanism matches the MSE of the discrete staircase \emph{including constants} as $\eps \to \infty$.

The GDL mechanism, as the difference of two i.i.d. negative binomial
noise shares, generalizes the distributed discrete Laplace mechanisms in \cite{goryczka2015comprehensive, balle2020private, bagdasaryan2022heatmaps} in the face of \emph{unexpected} dropouts when a larger fraction of parties than expected fail to add their noise shares,
providing a "smooth" closed-form privacy guarantee.

In \Cref{sec:integer-to-real}, we introduce a method to transform a discrete infinitely divisible additive mechanism into a continuous infinitely divisible mechanism up to a small loss in parameters. This approach allows us to achieve the asymptotically-optimal
$O\left(\Delta^2 e^{-2\eps/3}\right)$ MSE by transforming either the GDL or the MSDlap mechanisms, improving on the Arete mechanism's bound of $O\left(\Delta^2 e^{-\eps / 4}\right)$ in \cite{pagh2022infinitely}.

Our noise distributions and previously known distributions are summarized in \Cref{tab:summary}.

While its utility exceeds the GDL's, the MSDLap mechanism naively requires sampling from $O(\Delta)$ independent negative binomial random variables. In \Cref{sec:sampling} we outline an improved exact sampling algorithm which runs in only $O(1)$ steps in expectation for relevant regimes. This algorithm may be of independent interest for
general purpose multivariate negative binomial sampling in the sparse regime where most samples are 0.

Finally, in \Cref{sec:shuffle}, we improve the multi-message "split and mix" real summation shuffle protocol of \cite{balle2020private} with our results, attaining the $O\left(e^{-2\eps/3}\right)$ bound on MSE where previous results could only achieve
$O\left(1/\eps^2\right)$ from the discrete Laplace.

\section{Preliminaries}

For any $n \in \N$, we write $[n]$ as a shorthand for $\{1, \dots, n\}$.

\paragraph{Function identities.}
We introduce a few functions and identities used in our proofs. Let $\Gamma(z) = \int_{0}^\infty t^{z-1}e^{-t} \mathrm{d} t$ be the gamma function. We denote the rising factorial (aka the Pochhammer symbol) by %$(x)_n$, where
\begin{equation}\label{eq:pochhammer}
(x)_n := (x)(x + 1)\cdots(x + n -1) = \frac{\Gamma(x + n)}{\Gamma(x)}.
\end{equation}

Denote the hypergeometric function by $\HG{a}{b}{c}{z}$ \cite{DLMF-hyper}, where 
\begin{align}
    \HG{a}{b}{c}{z}
        &= \sum_{s=0}^\infty \frac{(a)_s(b)_s}{(c)_s s!}z^s \label{eq:hyperid1}\\
        &= \frac{\Gamma(c)}{\Gamma(a)\Gamma(b)}
            \sum_{s=0}^\infty
                \frac{\Gamma(a+s)\Gamma(b+s)}
                     {\Gamma(c + s) s!} z^s \label{eq:hyperid2}.
\end{align}

Let $\psi(z)$ denote the digamma function, where 
\begin{equation}\label{eq:digamma}
    \psi(z) = \frac{\mathrm{d}}{\mathrm{d} z} \log(\Gamma(z)) = \frac{\Gamma'(z)}{\Gamma(z)}
\end{equation}
%As the gamma function is log convex and increasing on $z > 0$, we make the following observation:
The following observation is well-known (see e.g.~\cite{alzer2017harmonic}):
\begin{observation}\label{obs:digamma}
    $\psi$ is increasing and $\psi'$ is decreasing on $(0, \infty)$.
\end{observation}

Finally, we state the so-called ``hockey-stick identity'' (e.g. \cite{jones1996generalized}):
\begin{lemma} \label{lem:pascal}
For any non-negative integers $\ell, m$, $\sum_{j=\ell}^{m} \binom{j}{\ell} = \binom{m+1}{\ell+1}$
\end{lemma}

\paragraph{Distributions.}
For %notational 
convenience, we write a random variable and its distribution interchangeably.

For a discrete distribution $\cD$ with support on $\cX$, denote its probability mass function (PMF) as $f_{\cD}(k)$ for $k \in \cX$. When we say that a discrete distribution is infinitely divisible, we assume implicitly that $\cD_{/n}$ are also discrete. Relevant to this work are the following well-known discrete distributions (where $a, r \in \R_{> 0}, p \in (0, 1)$):
\begin{itemize}

\item The negative binomial distribution, denoted $\NB(r, p)$ and with support on $\mathbb{Z}_{\ge 0}$, has PMF 
$
f_{\NB(r, p)}(k) = (1 - p)^k p^r \frac{\Gamma(k+r)}{\Gamma(r)\Gamma(k + 1)}.
$
 It is infinitely divisible, as $\sum_{i=1}^n \NB(r/n, p) \sim \NB(r, p)$. Its variance is $\Var[\NB(r, p)] = \frac{(1-p)r}{p^2}$.
For $r \in \mathbb{N}$, the negative binomial distribution models the number of failures before the first $r$ successes in a sequence of i.i.d. Bernoulli trials with success probability $p$.

\item The geometric distribution, denoted $\Geo(p)$ is a special case of the negative binomial with $r = 1$.

\item The discrete Laplace distribution, denoted $\DLap(a)$, has PMF 
$
f_{\DLap(a)}(k) = \tanh(a/2)e^{-a|k|}.
$ 
It is infinitely divisible, as 
$
\sum_{i=1}^n \left(\NB(1/n, 1 - e^{-a}) - \NB(1/n, 1 - e^{-a})\right) \sim \DLap(a)
$.
Its
variance is $\Var[\DLap(a)] = \frac{1}{\cosh(a) - 1}$.

\item The Bernoulli distribution, denoted $\Ber(p)$, has PMF 
$
f_{\Ber(p)}(k) = \begin{cases}
    p & k = 1\\
    1-p & k = 0
\end{cases}.
$
\end{itemize}

% charlie: we can't use \mathbb{R} as the support since we want to talk about the gamma distribution
For a continuous distribution $\cD$ on $\cX$, we use $f_{\cD}(x)$ for $x \in \cX$ to denote its probability density function (PDF) at $x$. We recall the following continuous distributions (where $k, \theta, b \in \R_{> 0}$):
\begin{itemize}
\item The gamma distribution, denoted by $\Gamma(k, \theta)$, with support on $\R_{+}$ has PDF 
$f_{\Gamma(k, \theta)}(x) = \frac{1}{\Gamma(k)\theta^k} \cdot x^{k-1}e^{-x/\theta}.$
It is infinitely divisible as $\sum_{i=1}^n \Gamma(k/n, \theta) \sim \Gamma(k, \theta)$.
\item The Laplace distribution, denoted by $\Lap(b)$, with support  on $\R$ has PDF 
$f_{\Lap(b)}(x) = \frac{1}{2b} e^{-|x|/b}.$
Its variance is $\Var[Lap(b)] = 2b^2$. It is infinitely divisible as  $
\sum_{i=1}^n \left(\Gamma(1/n, b) - \Gamma(1/n, b)\right) \sim \Lap(b).$
\end{itemize}

We will also use the following simple observations:
\begin{observation} \label{obs:scale-infdiv}
If a random variable $Z$ is infinitely divisible, then $c \cdot Z$ is infinitely divisible for any constant $c$.
\end{observation}
\begin{observation} \label{obs:sum-infdiv}
If random variables $Z_1, Z_2$ are infinitely divisible and independent, then $Z_1 + Z_2$ is infinitely divisible.
\end{observation}

We say that a distribution $\cD$ is \emph{closed under summation} if $\cD$ is infinitely divisible and additionally, $\cD_{/n}$ follows the same distribution family as $\cD$ for all $n \in \N$. This additional property provides benefits in the distributed setting as it ensures the mechanism's privacy is well-understood even as parties drop out or join the protocol.

\paragraph{Max Divergence and Differential Privacy.} We let $\dr{P}{Q}$ denote the max divergence between two distributions $P, Q$, i.e., $\sup_{x \in \supp(P)} \frac{f_{P}(x)}{f_Q(x)}$. We state the following well-known properties of $\dr{\cdot}{\cdot}$.
\begin{lemma}[Post-Processing] \label{lem:post-processing}
For any (possibly randomized) function $f$ and any random variables $U, V$, we have $\dr{f(U)}{f(V)} \leq \dr{U}{V}$.
\end{lemma}

\begin{lemma}[Triangle Inequality] \label{lem:tri-ineq}
For any distributions $P, Q, R$, $\dr{P}{Q} \leq \dr{P}{R} + \dr{R}{Q}$.
\end{lemma}

For a given query function $q: X^d \to \mathcal{Y}$, we let $\Delta(q) = \max_{x,x'} |q(x)-q(x')|$ where the maximum is over all pairs $x$ and $x'$ differing on one entry. The $\cD$-noise addition mechanism for a query function $q$ is the mechanism $M(x)$ that outputs $q(x) + Z$ where $Z$ is drawn from $\cD$. For a discrete (resp. continuous) distribution $\cD$ and $\Delta \in \N$ (resp. $\Delta \in \R_{> 0}$), we say the \emph{$\cD$-noise addition mechanism is $\eps$-DP for sensitivity $\Delta$} iff the $\cD$ noise addition mechanism is $\eps$-DP for all queries $q: \cX^d \to \Z$ (resp. $q: \cX^d \to \R$) such that $\Delta(q) \leq \Delta$. It follows from the definition of DP and max divergence that this condition translates to the following (see e.g. \cite{geng2014optimal}):
\begin{lemma} \label{lem:cond-dp}
For a discrete (resp. continuous) distribution $\cD$,  the $\cD$-noise addition mechanism is $\eps$-DP for sensitivity $\Delta$ iff $\dr{\cD + \xi}{\cD} \leq \eps$ for all $\xi \in \{-\Delta, -(\Delta - 1), \dots, \Delta\}$ (resp. if\footnote{The other direction of the implication for the continuous case does not hold since e.g. the set of $x$ where $f_{\cD + \xi}(x) > e^\eps \cdot f_{\cD}(x)$ might have measure zero.} $\dr{\cD + \xi}{\cD} \leq \eps$ for all $\xi \in [-\Delta, \Delta]$).
\end{lemma}

\section{Discrete mechanisms}\label{sec:discrete}
\begin{wrapfigure}{R}{0.39\textwidth}
  \includesvg[inkscapelatex=false, width =\linewidth]{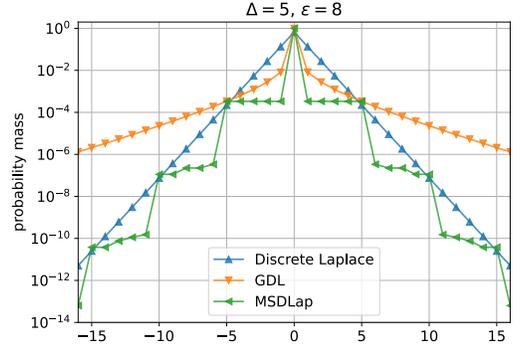}
  \caption{The PMF of the GDL distribution parameterized by \Cref{thm:opt-var}, the MSDLap distribution (\Cref{thm:multi-scale-dlap}), and the discrete Laplace distribution with $a = \eps/\Delta$. The GDL distribution has a much sharper peak around 0, before flattening out and decreasing slower than the discrete Laplace. The MSDLap has a "staircase" shaped distribution with sharp drops at $\Delta$-width intervals. Its PMF appears fully dominated by the discrete Laplace's,
  except at multiples of $\Delta$.}
  \label{fig:gdl}
  \vspace{-50pt}
\end{wrapfigure}

In this section, we present two infinitely divisible discrete additive-noise mechanisms that achieve order-optimal MSE.

\subsection{The generalized discrete Laplace mechanism}

This section introduces the generalized discrete Laplace distribution and associated mechanism. We start by the description of the distribution and its PMF:

\begin{definition}[\cite{LS14}]\label{def:gdl}
    For $\beta, a > 0$, let $\GDL(\beta, a)$ denote the distribution of $Z_1 - Z_2$ where  $Z_1, Z_2 \overset{\text{i.i.d.}}{\sim} \NB\left(\beta, 1 - e^{-a}\right)$.
    The PMF of $\GDL(\beta, a)$, i.e. $f_{\GDL(\beta, a)}(x)$, is equal to
    \begin{align}\label{eq:gdl-pmf}
    e^{-a|x|}\left(1-e^{-a}\right)^{2\beta} \HG{\beta}{\beta + |x|}{1 + |x|}{e^{-2a}}\frac{\Gamma(\beta + |x|)}{\Gamma(1 + |x|)\Gamma(\beta)}
    \end{align}
    for all $x \in \Z$.
\end{definition}

The discrete Laplace distribution is a special case of the GDL with $\beta = 1$, as
the PMF $f_{\GDL(1, a)}(k) = f_{\DLap(a)}(k) = \tanh(a/2)e^{-a|k|}$.

The infinite divisibility of the negative binomial distribution immediately implies that GDL is also infinitely divisible:

\begin{observation} \label{lem:inf-divis}
The GDL distribution is infinitely divisible and closed under summation. In particular, for independent %random variables
$
X_1 \sim \GDL(\beta_1, a), \cdots, X_n \sim \GDL(\beta_n, a),
$
we have $X = \sum_{i=1}^n X_i \sim \GDL(\sum_{i=1}^n \beta_i, a)$.
\end{observation}

To analyze the privacy guarantee of the GDL-noise addition mechanism, we will prove a few useful facts regarding its PMF. To do this, it is convenient to consider the extension of the PMF to all real numbers; let $f^\mathbb{R}_{\GDL(\beta, a)}: \R \to \R$ be the function defined by \eqref{eq:gdl-pmf}. We start with the following lemma.

\begin{lemma}\label{lem:gdl-logconvex}
For $\beta \in (0, 1)$, $f^\mathbb{R}_{\GDL(\beta, a)}$ is decreasing and log convex on $[0, \infty)$. % for $0 < \beta < 1$. %,
%where $f^\mathbb{R}_{GDL(\beta, a)}(x)$ denotes \Cref{eq:gdl-pmf} withc $k \in \mathbb{R}$, $\beta, a > 0$ i.e. it extends the GDL PMF to the reals. 
\end{lemma}

\begin{proof}
Because the product of log convex and decreasing functions
is log convex and decreasing, and $e^{-a x}$ is both log convex and decreasing, we focus on remaining relevant term.

\begin{align*}
\HG{\beta}{\beta + |x|}{1 + |x|}{e^{-2a}}&\frac{\Gamma(\beta + |x|)}{\Gamma(1 + |x|)\Gamma(\beta)} %\\
    %\text{\Cref{eq:hyperid2}}
    &\overset{\eqref{eq:hyperid2}}{=} 
    \sum_{s=0}^\infty \frac{\Gamma(\beta + s)\Gamma(\beta + x + s)}{\Gamma(\beta)^2\Gamma(1 + x + s)s!}e^{-2 a s} %\\&
    = \sum_{s=0}^\infty g(x) \frac{\Gamma(\beta +s) e^{-2 a s}}{\Gamma(\beta)^2 s!},
\end{align*}
where $g(x) = \frac{\Gamma(\beta + x + s)}{\Gamma(1 + x + s)}$. 
First we show that
$$g'(x) = g(x) (\psi(x + s + \beta) - \psi(x + s + 1)) < 0$$
and
\begin{align*}
\frac{\mathrm{d}^2}{\mathrm{d} x^2} \log g(x) &= \frac{\mathrm{d}^2}{\mathrm{d} x^2} \log(\Gamma(\beta + x + s)) - \log(\Gamma(1 + x + s))\\
&= \psi'(\beta + x + s) - \psi'(1 + x + s) \ge 0.
\end{align*}

The derivatives follow from \Cref{eq:digamma},
and the inequalities follow from \Cref{obs:digamma}. Finally, the result follows as the sum of decreasing and log convex functions is decreasing and log convex.
\end{proof}

%Before proceeding with our main privacy theorem we need to state a technical lemma.
We also need the following technical lemma which, together with \Cref{lem:gdl-logconvex}, allows us to consider only $f(0) / f(\Delta)$. Its proof is deferred to \Cref{app:logconvex-ratio-proof}.

\begin{lemma}\label{lem:logconvex-ratio}
Let $f: \R \to \R$ be symmetric about 0, and decreasing and log convex on $[0, \infty)$. Then for any $x, x'$ such that $|x-x'| \le \Delta$, we have
$
\frac{f(x)}{f(x')} \le \frac{f(0)}{f(\Delta)}.
$
\end{lemma}

We are now ready to state the privacy guarantee of the GDL-noise addition mechanism.

\begin{theorem} \label{thm:main-epsilon}
For any $\Delta \in \N, \beta, a > 0$, the $\GDL(\beta, a)$-noise addition mechanism is $\eps$-DP for sensitivity $\Delta$ iff
\begin{align*}
\eps \leq 
\begin{cases}
a \Delta + \log \frac {\HG{\beta}{\beta}{1}{e^{-2a}}}
    {\HG{\beta}{\beta + \Delta}{1 +\Delta}{e^{-2a}}}
    \frac{\Gamma(\Delta+1)\Gamma(\beta)}{\Gamma(\beta + \Delta)}
    & 0 < \beta < 1\\
a \Delta & \beta \ge 1
\end{cases}
\end{align*}
\end{theorem}

\begin{proof}
For the $0 < \beta < 1$ case, by \Cref{lem:cond-dp,lem:logconvex-ratio} along with the fact that $f_{\GDL(\beta, a)}^{\R}$ is log-convex on $[0, \infty)$ and symmetric about 0, the $\GDL(\beta, a)$-noise addition mechanism is $\eps$-DP for sensitivity $\Delta$ iff
\begin{align*}
\eps \leq \max_{\xi \in \{-\Delta, \dots, \Delta\}} \max_{x \in \Z} \frac{f_{\GDL(\beta, a)}(x - \xi)}{f_{\GDL(\beta, a)}(x)} \overset{\text{\Cref{lem:logconvex-ratio}}}{=} \log \frac{f_{\GDL(\beta, a)}(0)}{f_{\GDL(\beta, a)}(\Delta)},
\end{align*}
which is exactly the claimed bound.

For the $\beta \ge 1$ case, we decompose the mechanism by observing (see \Cref{lem:inf-divis}) that $Z \sim \GDL(\beta, a)$ can be sampled as $Z = Z_1 + Z_2$ where $Z_1 \sim \DLap(a), Z_2 \sim \GDL(\beta - 1, a)$ are independent. By post-processing (\Cref{lem:post-processing}), this means that the $\GDL(\beta, a)$-noise addition mechanism is at least as private as the the
$\DLap(a)$-noise addition mechanism, which satisfies $a \Delta$-DP for sensitivity $\Delta$.

For the tightness claim, first note that \Cref{eq:hyperid1} yields
\begin{align*}
&\log\left(\frac{f_{\GDL(\beta, a)}(k)}{f_{\GDL(\beta, a)}(k + \Delta)}\right) %\\ &
= a \Delta + 
    \log \left(
        \frac
            {\sum_{s=0}^\infty \frac{(\beta)_s (\beta + k)_s e^{-2 a s}}{(1 + k)_s s!}}
            {\sum_{s=0}^\infty \frac{(\beta)_s (\beta + k + \Delta)_s e^{-2 a s}}{(1 + k + \Delta)_s s!}} 
        \cdot
        \frac{\Gamma(\beta + k) \Gamma(1 + k + \Delta)}
             {\Gamma(1 + k)\Gamma(\beta + k + \Delta)}
    \right).
\end{align*}
Observe that
\begin{align*}
\frac{\frac{(\beta)_s (\beta + k)_s e^{-2 a s}}{(1 + k)_s s!} }{\frac{(\beta)_s (\beta + k + \Delta)_s e^{-2 a s}}{(1 + k + \Delta)_s s!}} = \prod_{i=0}^{s-1}\frac{\frac{\beta + k + i}{1 + k + i}}{\frac{\beta + k + \Delta + i}{1 + k + \Delta + i}} \geq 1,
\end{align*}
%where the inequality holds 
because every term in the product is at least one. Plugging this into the above, we get
\begin{align*}
\log\left(\frac{f_{\GDL(\beta, a)}(k)}{f_{\GDL(\beta, a)}(k + \Delta)}\right) %&
\geq a \Delta + 
    \log \left(
        \frac{\Gamma(\beta + k) \Gamma(1 + k + \Delta)}
             {\Gamma(1 + k)\Gamma(\beta + k + \Delta)}
    \right) %\\ &
    = a\Delta + \log\left(\frac{(1 + k)_{\Delta}}{(\beta + k)_{\Delta}}\right).
\end{align*}
Since $\lim_{k \to \infty} \frac{(1 + k)_{\Delta}}{(\beta + k)_{\Delta}} = 1$, $\lim_{k \to \infty} \log\left(\frac{f_{\GDL(\beta, a)}(k)}{f_{\GDL(\beta, a)}(k + \Delta)}\right) \geq a\Delta,$ meaning that our bound is tight.
%
%The tightness of
%the upper bound can be seen by taking 
%\modified{
%\begin{align*}
%&\lim_{k \to \infty} \log\left(\frac{f_{\GDL(\beta, a)}(k)}{f_{\GDL(\beta, a)}(k + \Delta)}\right)  \\
%    &= a \Delta + 
%    \log \left(
%        \lim_{k \to \infty}
%        \frac
%            {\sum_{s=0}^\infty \frac{(\beta)_s (\beta + k)_s e^{-2 a s}}{(1 + k)_s s!}}
%            {\sum_{s=0}^\infty \frac{(\beta)_s (\beta + k + \Delta)_s e^{-2 a s}}{(1 + k + \Delta)_s s!}} 
%        \cdot
%        \frac{\Gamma(\beta + k) \Gamma(1 + k + \Delta)}
%             {\Gamma(1 + k)\Gamma(\beta + k + \Delta)}
%    \right)\\
%&= 
%    a \Delta + 
%    \log \left(
%        \frac
%            {\sum_{s=0}^\infty \frac{(\beta)_s }{s!}e^{-2 a s}}
%            {\sum_{s=0}^\infty \frac{(\beta)_s }{s!}e^{-2 a s}}
%    \right) + 
%    \log \left(
%        \lim_{k \to \infty}
%        \frac{\Gamma(1 + k + \Delta) \Gamma(\beta + k)}
 %            {\Gamma(1 + k)\Gamma(\beta + k + \Delta)}
%    \right)\\
%&= a \Delta + 
%    \log\left(\lim_{k\to\infty} \frac{(1 + k)_\Delta}{(\beta + k)_\Delta}\right)
%= a \Delta,
%\end{align*}
%where the first equality uses the identity in \Cref{eq:hyperid1}. 
%Both sub limit evaluations follow from \pasin{I'm pretty sure you can't do this because the summation of $s$ is from 0 to $\infty$ whereas the following assume that $s$ is fixed and $k \to \infty$.}
%\begin{align*}
%\lim_{k\to\infty} \frac{(k+x)_s}{(k + y)_s} = 
%\left(\lim_{k\to \infty} \frac{k+x}{k+y}\right) \cdots \left(\lim_{k\to \infty} \frac{k + x + s - 1}{k+y+s-1}\right) = 1.&\qedhere
%\end{align*}
%}
\end{proof}

The exact privacy bound for the GDL above (in the case $0 < \beta < 1$) is fairly unwieldy, and implementations of hypergeometric and gamma functions can quickly lead to numerical issues in
practice, especially for large\footnote{
$\Gamma(x+1)$ will overflow a 64 bit unsigned integer for $x \ge 21$.} values of $\Delta$. Therefore, we will show a simplified privacy bound below,  as well as a tighter, slightly more complex bound in \Cref{sec:alt-gdl-privacy}.
\begin{corollary}\label{cor:hyper-ratio}
For any $\Delta \in N, a > 0, \beta \in (0, 1)$, the $\GDL(\beta, a)$-noise addition mechanism is $\eps$-DP for sensitivity $\Delta$ where $\eps \le a \Delta + \log \frac{\Delta}{\beta}$.
\end{corollary}
\begin{proof}

We begin by bounding the ratio of the hypergeometric functions above by 1, i.e.,
$$
\frac  {\HG{\beta}{\beta}{1}{e^{-2a}}}
       {\HG{\beta}{\beta + \Delta}{1 +\Delta}{e^{-2a}}}
= \frac
    {\sum_{s=0}^\infty \frac{(\beta)_s (\beta)_s}{s! \cdot s!} e^{-2 a s}}
    {\sum_{s=0}^\infty \frac{(\beta)_s (\beta + \Delta)_s}{(1+\Delta)_s s!} e^{-2 a s}}
\le 1.
$$

This follows from showing the inequality term wise: %that $\forall s\ge 0 \in \mathbb{Z}$ and $\Delta \ge 1 \in \mathbb{Z}$
\begin{align*}
\frac{(\beta)_s}{s!} &= 
\left(\frac{\beta}{1}\right)
\left(\frac{\beta + 1}{2}\right)\cdots
\left(\frac{\beta + s - 1}{s}\right)  %\\
%&
\le
\left(\frac{\beta + \Delta}{1 + \Delta}\right)
\left(\frac{\beta + 1 + \Delta}{2 + \Delta}\right)\cdots
\left(\frac{\beta + s - 1 + \Delta}{s + \Delta}\right) %\\
&= \frac{(\beta + \Delta)_s}{(1 + \Delta)_s}.
\end{align*}

Therefore, from \Cref{thm:main-epsilon}, the mechanism is $\eps$-DP for all $\eps \le a \Delta + \log \frac{\Gamma(\Delta + 1)\Gamma(\beta)}{\Gamma(\beta + \Delta)}$. Finally, observe that
$
\frac{\Gamma(\Delta + 1)\Gamma(\beta)}{\Gamma(\beta + \Delta)} 
= \frac{\Delta !}{(\beta)_\Delta} 
= \frac{\Delta}{\beta} \cdot \frac{(\Delta - 1)!}{(\beta + 1)_{\Delta-1}}
\leq \frac{\Delta}{\beta}.
$
Combining the two inequalities yields the desired claim.
\end{proof}

Finally, we prove our main accuracy theorem about the GDL mechanism in the low-privacy regime.

\begin{theorem}\label{thm:opt-var}
For any $\Delta \in \N$ and $\eps > 2 + \log \Delta$, the $\GDL\left(\Delta e^{2-\eps}, \frac{2}{\Delta}\right)$-noise addition mechanism is $\eps$-DP for sensitivity $\Delta$ 
and has MSE $O(\Delta^3 e^{-\eps})$. 
\end{theorem}

\begin{proof}
\emph{(Privacy)} Since $\eps > 2 + \log \Delta$, we have $\beta = \Delta e^{2-\eps} < 1$. Thus, \Cref{cor:hyper-ratio} immediately implies the privacy guarantee.

\emph{(Accuracy)} The MSE of the mechanism is
\begin{align*}
\Var(\GDL(\beta, a)) = 2 \Var(\NB\left(\beta, 1 - e^{-a}\right))
= \frac{\beta}{\cosh(a) - 1}
= \frac{\Delta e^{2-\eps}}{\cosh(2/\Delta) - 1}
= O(\Delta^3 e^{-\eps}). &\qedhere
\end{align*}
\end{proof}

\subsection{The multi-scale discrete Laplace mechanism}
The \emph{$(\eps, \Delta)$-multi-scale discrete Laplace ($(\eps, \Delta)$-MSDLap) distribution} with parameter $\eps > 0, \Delta \in \N$ is defined as the distribution of $\sum_{i=1}^{\Delta} i \cdot X_i$ where $X_1, \dots, X_{\Delta} \overset{\text{i.i.d.}}{\sim} \DLap(\eps)$. From \Cref{obs:scale-infdiv} and \Cref{obs:sum-infdiv}, the $(\eps, \Delta)$-MSDLap distribution is infinitely divisible.

We give a proof below that this mechanism is $\eps$-DP, and that its accuracy guarantee matches that of \Cref{thm:opt-var}.

\begin{theorem} \label{thm:multi-scale-dlap}
For any $\eps > 0, \Delta \in \N$, the $(\eps, \Delta)$-MSDLap-noise addition mechanism is $\eps$-DP for sensitivity $\Delta$. Furthermore, for $\eps \geq 1$, the MSE is $O(\Delta^3 \cdot e^{-\eps})$.
\end{theorem}

\begin{proof}
\emph{(Privacy)} From \Cref{lem:cond-dp} and the symmetry of the noise around 0, it suffices to show \\
$\dr{\xi + \sum_{i=1}^{\Delta} i \cdot X_i}{\sum_{i=1}^{\Delta} i \cdot X_i} \leq \eps$ for all $\xi \in [\Delta]$. From \Cref{lem:post-processing}, we have
\begin{align*}
\dr{\xi + \sum_{i=1}^{\Delta} i \cdot X_i}{\sum_{i=1}^{\Delta} i \cdot X_i} \leq \dr{\xi + \xi \cdot X_{\xi}}{\xi \cdot X_{\xi}} = \dr{1 + X_{\xi}}{X_{\xi}} \leq \eps,
\end{align*}
where the last inequality follows from $X_i \sim \DLap(\eps)$. Thus, the $(\eps, \Delta)$-MSDLap-noise addition mechanism is $\eps$-DP for sensitivity $\Delta$.
%Consider any two neighboring datasets $x, x'$. We will show that $\eps \geq \dr{q(x) + \sum_{i=1}^{\Delta} X_i}{q(x') + \sum_{i=1}^{\Delta} X_i} = \dr{q(x) - q(x') + \sum_{i=1}^{\Delta} X_i}{\sum_{i=1}^{\Delta} X_i}$. If $q(x) = q(x')$, then this is obvious. Otherwise, let $i^* = |q(x) - q(x')|$. From \Cref{lem:post-processing}, we have
%\begin{align*}
%&\dr{q(x) - q(x') + \sum_{i=1}^{\Delta} X_i}{\sum_{i=1}^{\Delta} X_i} \\
%&\leq \dr{q(x) - q(x') + i^* \cdot X_i}{i^* \cdot X_i} \\
%&= \dr{1 + X_i}{X_i} \\
%&\leq \eps,
%\end{align*}

\emph{(Accuracy)} The MSE is
$
\Var\left(\sum_{i=1}^{\Delta} i \cdot X_i\right) 
%&= \sum_{i=1}^{\Delta} \Var(i \cdot X_i)\\
=  \sum_{i=1}^{\Delta} i^2 \cdot \Var(X_i) 
= \frac{\Delta(\Delta + 1)(2\Delta + 1)}{6 (\cosh(\eps) - 1)} 
= O(\Delta^3 \cdot e^{-\eps}). %& \qedhere
$
\end{proof}

The above theorem implies that, for fixed $\Delta$ and sufficiently large $\eps$, the $(\eps, \Delta)$-MSDLap-noise achieves asymptotically optimal MSE. In fact, below we prove a stronger statement that, if we take $\eps \to \infty$, the ratio between MSEs of MSDLap and the discrete staircase mechanism~\cite{stair2015,stair2016adding} approaches 1. 

\begin{corollary}\label{cor:msdlap-matches-staircase}
For a fixed $\Delta$, the ratio of the MSE of the $(\eps, \Delta)$-MSDLap-noise addition mechanism and that of the $\eps$-DP discrete staircase mechanism for sensitivity $\Delta$~\cite{stair2015,stair2016adding} approaches 1 as $\eps \to \infty$.
%    For a fixed $\Delta$ and arbitrarily large $\eps$, the MSE of the multi-scale discrete Laplace mechanism matches (including constants) the discrete staircase mechanism with $r=1$ of $\frac{1}{3}e^{-\eps} \Delta (\Delta + 1)(2\Delta + 1)$
\end{corollary}
\begin{proof}
As we show in \Cref{obs:dstair-var-large} (in \Cref{sec:dstair-var}), the MSE of the $\eps$-DP discrete staircase mechanism for sensitivity $\Delta$ is at least $\frac{1}{e^{\eps} + (2\Delta - 1)} \cdot \frac{\Delta (\Delta + 1)(2\Delta+1)}{3}$ for any sufficiently large $\eps$. Thus, in this regime, the ratio between the two MSEs is at most
$
\frac{\frac{\Delta(\Delta + 1)(2\Delta + 1)}{6 (\cosh(\eps) - 1)}}{\frac{1}{e^{\eps} + (2\Delta - 1)} \cdot \frac{\Delta (\Delta + 1)(2\Delta+1)}{3}} = \frac{1 + (2\Delta - 1)e^{-\eps}}{1 - 2e^{-\eps} + e^{-2\eps}}.
$
The RHS approaches 1 as $\eps \to \infty$ as claimed.
\end{proof}

While the naive approach to sample from the MSDLap mechanism
requires sampling from $O(\Delta)$ random variables, we show in \Cref{sec:sampling} an efficient algorithm tuned for the high $\eps$ regime. 
% Pasin: I'm commenting the following sentence out. Even without Section 5, I think this can still be easily implemented in practice for say Delta <= 20. I'm actually not sure how often you need incredibly large Delta (say >= 1000) in practice. 
%, \modified{ensuring the benefits of the MSDLap mechanism can be realized in practice.}
\subsubsection{Generalizing the MSDLap mechanism}

We can also generalize the multi-scale discrete Laplace mechanism to match the error in~\cite{geng2014optimal} for every setting of parameters $\Delta, \eps$. We state the theorem below where $r \in \{0, \dots, \Delta\}$ is the free parameter so that it matches the ``$r$'' parameter in the discrete staircase mechanism as presented in \cite{geng2014optimal}. Note that in our results henceforth, we sometimes assume that $\eps \geq 2$ for simplicity; the constant 2 can be changed to any constant\footnote{Namely, by changing the distribution of $X^*$ in the proof of \Cref{thm:all-regime-dlap} to $\DLap(c/r)$ for some larger constant $c > 1$.} but we keep it for simplicity of the distribution description.

\begin{theorem} \label{thm:all-regime-dlap}
For any $\eps \geq 2, \Delta \in \N$ and every $r \in \{0, \dots, \Delta\}$, there exists an infinitely divisible discrete noise-addition mechanism that is $\eps$-DP for sensitivity $\Delta$ and has MSE $O\left(r^2 + \frac{e^{-\eps}\Delta^3}{r + 1}\right)$. 
\end{theorem}

%It should be noted that, 
By plugging in $r = 0, r = \lceil e^{-\eps/3} \Delta \rceil$, we get the following, which will be useful later on in \Cref{sec:shuffle}.
\begin{corollary}\label{cor:msdlap-rextremes}
For any $\eps \geq 2$, there exists an infinitely divisible discrete noise-addition mechanism that is $\eps$-DP for sensitivity $\Delta$ with MSE $O\left(\Delta^2 \min\{e^{-\eps} \Delta, e^{-2\eps/3}\}\right)$. 
\end{corollary}

We are now ready to prove \Cref{thm:all-regime-dlap}. The rough idea is that, instead of using only the $(\eps, \Delta)$-MSDLap noise, we first add a scaled-up $(\eps - 1, \Delta_0)$-MSDLap noise where $\Delta_0 < \Delta$. The scaling up leaves us with ``holes'' in the noise distribution. To fix this, we additionally add another DLap noise to ``smoothen out the holes''. This idea is formalized below.

\begin{proof}[Proof of \Cref{thm:all-regime-dlap}]
The case $r = 0$ follows from \Cref{thm:multi-scale-dlap}. Thus, we can henceforth consider $r \geq 1$.

Let $\Delta_0 = \lfloor \Delta / r \rfloor$. Let the noise distribution $\cD$ be the distribution of $Z = r \cdot X + Y$ where $X \sim (\eps-1,\Delta_0)\mathrm{-MSDLap}$ and $Y \sim \DLap(1/r)$ are independent. The infinite divisibility of $\cD$ follows from the infinite divisibility of MSDLap, DLap and \Cref{obs:scale-infdiv,obs:sum-infdiv}.

\emph{(Privacy)} 
From \Cref{lem:cond-dp}, it suffices to show $\dr{\xi + Z}{Z} \leq \eps$ for all $\xi \in \{-\Delta, \dots, \Delta\}$. Let $i^* = \lfloor \xi / r \rfloor$ and $j^* = \xi - r \cdot i^*$. Note that $i^* \in [\Delta_0]$ and $j^* \in [r]$. From \Cref{lem:tri-ineq,lem:post-processing}, we then have
\begin{align*}
\dr{\xi + Z}{Z}
&= \dr{r \cdot i^* + j^* + r \cdot X + Y}{r \cdot X + Y} \\
&\leq \dr{r \cdot i^* + j^* + r \cdot X + Y}{r \cdot i^* + r \cdot X + Y} + \dr{r \cdot i^* + r \cdot X + Y}{r \cdot X + Y} \\
&\leq \dr{j^* + Y}{Y} + \dr{i^* + X}{X} \\
&\leq 1 + (\eps - 1) = \eps,
\end{align*}
where the last inequality follows from $Y \sim \DLap(1/r)$ and $X \sim (\eps - 1, \Delta_0)\textrm{-}\MSDLap$ (together with \Cref{thm:multi-scale-dlap} and \Cref{lem:cond-dp}).

\emph{(Accuracy)} The MSE is
$
%\Var\left(r \cdot X + Y\right)  = 
r^2 \cdot \Var(X) + \Var(Y)
\leq O(r^2 \cdot \Delta_0^3 \cdot e^{-\eps}) +  O(r^2) 
= O(r^2 + e^{-\eps} \Delta^3 / r)$
\end{proof}

\begin{remark} \label{remark:generalize}
The MSDLap mechanism can be generalized further than in the proof of \Cref{thm:all-regime-dlap}. In particular, rather than only considering $\DLap$ noise as our basic primitive, we can consider the $(\cD_1, \cD_2)$-generalized-multi-scale mechanism that adds noise from two different distributions $\cD_1$ and $\cD_2$ where
\begin{itemize}
    \item the $\cD_1$-noise addition mechanism is $\eps_1$-DP for $\Delta = 1$, and is added at multiple scales
    \item the $\cD_2$-noise addition mechanism is $\eps_2$-DP for $\Delta = r$, and is used to "smoothen out the holes" in the noise from $\cD_1$
\end{itemize}
More specifically, the final noise we add is $r \cdot \left(\sum_{i=1}^{\Delta_0} i \cdot X_i\right) + Y$ where $X_1, \dots, X_{\Delta_0} \overset{\text{i.i.d.}}{\sim} \cD_1$ and $Y \sim \cD_2$ are independent. The proof of privacy in this case proceeds identically to \Cref{thm:all-regime-dlap} yielding an $(\eps_1 + \eps_2)$-DP mechanism overall.
% Pasin: I'm commenting the following sentence out, I think just mentioning GDL & closeness under summation is enough.
%We note that this allows us to consider the distributed MSDLap mechanism where only a $\beta$ fraction of parties add their noise. 
This allows us to consider $\cD_1, \cD_2 \sim \GDL$, which gives us the multi-scale version of the GDL mechanism; this noise distribution is additionally closed under summation.
%If $\cD_1 \sim \GDL(\beta, a_1)$ and $\cD_2 \sim \GDL(\beta, a_2 / r)$, then the $(\cD_1, \cD_2)$-multi scale mechanism is $(a_1 + a_2 + \log(r) - 2\log(\beta))$-DP from \Cref{cor:hyper-ratio}.
\end{remark}

We conclude this section by plotting the MSE of our new mechanisms in \Cref{fig:discrete} vs. %established
baselines.

\section{From Integer-Valued to Real-Valued Functions}\label{sec:integer-to-real}

We next show simple methods to transform a mechanism for integer-valued functions to real-valued functions. The approach is similar to that of \Cref{thm:all-regime-dlap}, except that we use the (continuous) Laplace noise to ``smoothen out the holes'' instead of its discrete analogue.

\begin{theorem}\label{thm:discrete-to-cont}
For $\eps \geq 2$ and $\Delta > 0$, there exists a continuous infinitely divisible noise-addition mechanism that is $\eps$-DP for sensitivity $\Delta$ with MSE $O(\Delta^2 \cdot e^{-2\eps/3})$. 
\end{theorem}

\begin{proof}
Since we can scale the input by $1/\Delta$, add noise and rescale back, we assume w.l.o.g. that $\Delta = 1$. 

Let $\eps_{\disc} = \eps - 1, \Delta_{\disc} = \lceil e^{\eps/3} \rceil$ and $r = \frac{1}{\Delta_{\disc}}$. Let $\cD_{\disc}$ be any infinitely divisible discrete distribution such that the $\cD_{\disc}$-noise addition mechanism is $\eps_{\disc}$-DP for sensitivity $\Delta_{\disc}$ with MSE $O(\Delta_{\disc}^3 \cdot e^{-{\eps_{\disc}}})$. (Such a distribution exists due to \Cref{thm:multi-scale-dlap}.\footnote{The GDL distribution also meets the requirements, but only for a specific regime of $\eps$. Specifically, from \Cref{thm:opt-var} we need to ensure that $\eps_\disc > 2 + \log(\Delta_\disc) > 2 + \eps / 3$. For $\eps_\disc=\eps - 1$, the continuous transformation of the GDL is valid when $\eps > 4.5$.})

Let $\cD$ be the distribution of $Z = r \cdot X + Y$ where $X \sim \cD_{\disc}$ and $Y \sim \Lap\left(r/2\right)$ are independent. Since both $X, Y$ are infinitely divisible, \cref{obs:scale-infdiv,obs:sum-infdiv} imply that $\cD$ is also infinitely divisible.

\emph{(Privacy)} From \Cref{lem:cond-dp}, it suffices to show  $\dr{\xi + Z}{Z} \leq \eps$ for all $\xi \in [-1, 1]$. Let $i^*$ be the closest integer to $\xi / r$ and let $j^* = \xi - r \cdot i^*$. Note that $i^* \in \{-\Delta_{\disc}, \dots, \Delta_{\disc}\}$ and\footnote{
Note that the choice of $i^*$ in order to halve the maximum value of $|j^*|$ cannot be directly applied to \Cref{thm:all-regime-dlap}. In that theorem,
we take $\Delta_0 = \lfloor \Delta / r \rfloor$, while the approach here only works when $\Delta_0 \ge \Delta / r$.} $j^* \in [-r/2, r/2]$. 
From \Cref{lem:tri-ineq,lem:post-processing}, we have
\begin{align*}
\dr{\xi + Z}{Z}
&= \dr{r \cdot i^* + j^* + r \cdot X + Y}{r \cdot X + Y} \\
&\leq \dr{r \cdot i^* + j^* + r \cdot X + Y}{r \cdot i^* + r \cdot X + Y} + \dr{r \cdot i^* + r \cdot X + Y}{r \cdot X + Y} \\
&\leq \dr{j^* + Y}{Y} + \dr{i^* + X}{X} \\
&\leq 1 + (\eps - 1) = \eps,
\end{align*}
where the last inequality follows from $Y \sim \Lap(r/2)$ and $X$-noise addition mechanism is $\eps$-DP for sensitivity $\Delta_{\disc}$ (and \Cref{lem:cond-dp}).

\emph{(Accuracy)} The MSE of the mechanism is
\begin{align*}
\Var\left(r \cdot X + Y\right) 
= r^2 \cdot \Var(X) + \Var(Y)
\leq O\left(\Delta_{\disc} \cdot e^{-\eps}\right) + O\left(\left(\frac{1}{\Delta_{\disc}}\right)^2\right)
\leq O(e^{-2\eps/3}),
\end{align*}
where the last inequality is from our choice $\Delta_{\disc} = \Theta(e^{\eps/3})$.
\end{proof}

\modified{Similar to \Cref{remark:generalize}, we may also consider $\cD_{\disc}$ apart from the MSDLap distribution, as long as it satisfies $\eps_{\disc}$-DP for sensitivity $\Delta_{\disc}$. In particular, we may use the GDL distribution.}
We conclude this section by plotting the results of transforming the discrete mechanisms from \Cref{sec:discrete} in \Cref{fig:cont}.

\begin{figure}[t]
  \centering
  \includesvg[inkscapelatex=false, width =1\linewidth]{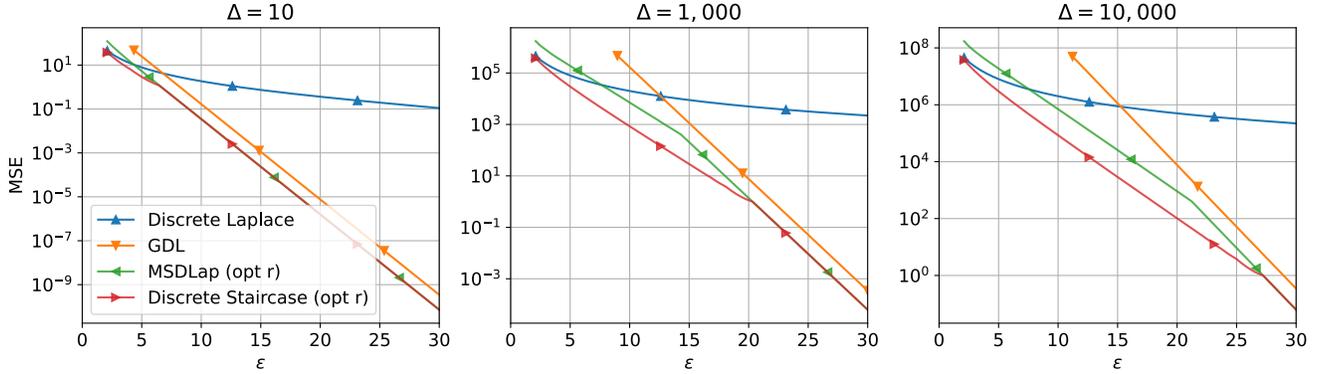}
  \caption{The MSE of the GDL mechanism (\Cref{thm:opt-var}) and  MSDLap mechanism (\Cref{thm:multi-scale-dlap}) with optimized $r$. We include the discrete Laplace and staircase (\Cref{sec:dstair-var}) baselines. In the high $\eps$ regime our mechanisms closely track the MSE of the discrete staircase. The MSDLap mechanism meets the MSE exactly at high $\eps$.}
  \label{fig:discrete}
\end{figure}

\begin{wrapfigure}{R}{0.45\textwidth}
  \includesvg[inkscapelatex=false, width = \linewidth]{images/mse-cont.svg}
  \caption{The MSE of the continuous GDL (\Cref{thm:opt-var}) and MSDLap (\Cref{thm:multi-scale-dlap}) after the continuous transformation of \Cref{thm:discrete-to-cont} is applied to them. We also plot the Arete \cite[Lemma 3]{pagh2022infinitely} and continuous staircase \cite[eq. 51]{geng2014optimal} mechanisms as baselines.}
  \label{fig:cont}
  \vspace{-28pt}
\end{wrapfigure}

\section{Efficiently and exactly sampling from the distributed MSDLap}\label{sec:sampling}
This section outlines an algorithm to efficiently sample from the MSDLap mechanism described in \Cref{thm:multi-scale-dlap}
in the distributed setting over $n$ parties. Recall that the $(\eps, \Delta)$-MSDLap noise is defined as $\sum_{i=1}^{\Delta} i \cdot X_i$ where $X_1, \dots, X_\Delta \overset{i.i.d.}{\sim} \DLap(\eps) = \NB(1, 1-e^{-\eps}) - \NB(1, 1-e^{-\eps})$. In other words, each of the $n$ parties need to sample $\sum_{i=1}^{\Delta} i \cdot (U_i - V_i)$ where $U_1, \dots, U_{\Delta}, V_1, \dots, V_{\Delta} \overset{i.i.d.}{\sim} \NB(1/n, 1-e^{-\eps})$. Thus, the naive algorithm %over $n$ parties 
requires each party to sample from $k = 2 \Delta$ negative binomial random variables. %: $Z \sim \sum_{i=1}^\Delta i \cdot (\NB(1/n, 1-e^{-\eps}) - \NB(1/n, 1 - e^{-\eps}))$. 
For large $\Delta$, or for $\Delta = O(e^\eps/3)$ as in \Cref{cor:msdlap-rextremes}, this may be computationally expensive.

Our algorithm resolves this issue, allowing us to sample from exponentially many (in $\eps$) negative binomial random variables in time polynomial in $\eps$.

%\begin{corollary}[to \Cref{thm:multi-nb-%sample}]\label{cor:constant-time-sampler}
%    Let $\gamma > 0 \in \mathbb{Q}$ and $\eps = \log\left(\frac{1}{1-e^{-\gamma}}\right)$. \Cref{alg:multi-nb-sample} can sample from $k = O(e^{\eps} / r)$ i.i.d. $\NB(r, 1-e^{-\eps})$ random variables using $O(1)$ arithmetic steps in expectation wherever $\eps > 1$.
%\end{corollary}

\begin{theorem}\label{thm:multi-nb-sample}
For input $k \in \mathbb{N}, r, \gamma \in \mathbb{Q}_{> 0}$, and $p = e^{-\gamma}$, the procedure described in \Cref{alg:multi-nb-sample} returns non-zero samples from $k$ i.i.d. samples of $\NB(r, 1-e^{-\eps})$ and completes in $\widetilde{O}\left(1 + k \cdot r \cdot e^{-\eps}\right)$ steps in expectation, where $\eps = \log \frac{1}{1-e^{-\gamma}}$ and $\widetilde{O}$ hides polynomial factors in $\eps$.
\end{theorem}

Our approach leverages the fact that in the high $\eps$ regime, most of these negative binomial random variables will be 0. We re-frame the problem of sampling many negative binomials into two separate problems:

\begin{itemize}
    \item Sampling from the \emph{sum} of many i.i.d. negative binomials.
    \item Fairly \emph{allocating} the result across each random variable.
\end{itemize}

While sampling from the sum of many negative binomials is simple on its face given their infinite divisibility, standard exact samplers for $\NB(r, p)$ (e.g. \cite{heaukulani2019black}) take time linear in $r$ which is not desirable. Below we will describe an algorithm (\Cref{alg:nb-sample-geo}) whose (expected) running time only scales with the mean of $\NB(r, p)$, which is only $O(r \cdot e^{-\eps})$ in our setting.

To fairly allocate across the random variables, we leverage the fact that the conditional distribution
of the sequence of $\NB$ random variables given their sum follows the
\emph{Dirichlet multinomial} distribution\footnote{Like the negative binomial with real stopping parameter, the Dirichlet multinomial distribution is \emph{also} sometimes called the (multivariate) P\'{o}lya distribution.} denoted $\DirM(n ,\boldsymbol{\alpha})$ where $\boldsymbol{\alpha} = \{\alpha_1, \dots, \alpha_k\}$ and $\alpha_0 = \sum_{i=1}^k \alpha_i$.

\begin{align*}
f_{\DirM(n, \boldsymbol{\alpha})}(x) &=
    \frac{\Gamma(\alpha_0)\Gamma(n+1)}{\Gamma(n + \alpha_0)}
    \prod_{i=1}^k \frac{\Gamma(x_i + \alpha_i)}{\Gamma(\alpha_i)\Gamma(x_i + 1)}
\end{align*}

\begin{lemma}[e.g. \cite{zhou2018nonparametric,townes2020review}]\label{lem:nb-conditioned-sum-is-dirm}
Let $X = \{X_1, \dots, X_k\}$ be a vector of independent entries where each $X_i \sim \NB(\alpha_i, p)$. Let $T = \sum_{i=1}^k X_i$. Then the conditional distribution $X | T = t \sim \DirM(t, \boldsymbol{\alpha})$. 
\end{lemma}
\begin{proof}
The result follows from Bayes' theorem.
\begin{align*}
f_{X | T = t}(x) &= \frac{f_{T | X = x}(t)f_{X}(x)}{f_T(t)}
 = \frac{\prod_{i=1}^k f_{\NB(\alpha_i, p)}(x_i)}
        {f_{\NB(\alpha_0, p)}(t)}\\
&=  \frac{\Gamma(\alpha_0)\Gamma(t + 1)}{(1-p)^{t}p^{\alpha_0}\Gamma(t + \alpha_0)}
    \prod_{i=1}^k 
        \frac{(1-p)^{x_i}p^{\alpha_i}\Gamma(x_i + \alpha_i)}{\Gamma(\alpha_i)\Gamma(x_i + 1)}\\
&= f_{\DirM(t, \boldsymbol{\alpha})}(x) &\qedhere
\end{align*}
\end{proof}

Our sampler can be implemented on a finite computer in the Word RAM model avoiding any real-arithmetic operations. Following  \cite[Section 5]{canonne2020discrete}, we focus on the runtime in the \emph{expected} number of arithmetic operations, which take only polynomial time in the bit complexity of the parameters. We make the following assumptions on the availability of basic sampling primitives requiring only $O(1)$ arithmetic operations in expectation:
\begin{itemize}
    \item A uniform sampler to draw $D \in \{1, 2, \dots, d\}$ for $d \in \mathbb{N}$.
    \item A $\Ber(n/d)$ sampler for $n,d\in \mathbb{N}$, which can be trivially implemented using the uniform sampler.
    \item A $\Geo(1-e^{-\gamma})$ sampler for $\gamma \in \mathbb{Q}$ from \cite{canonne2020discrete}.
\end{itemize}

Finally, we assume we have access to a map data structure with accesses and updates requiring $O(1)$ operations in expectation, and a vector data structure with $O(1)$ random access and append operations.

Our \Cref{alg:multi-nb-sample} is straightforward. It consists of
\begin{enumerate}
    \item Sampling from the negative binomial distribution with \Cref{alg:nb-sample-geo} to learn the sum of all the terms, handling rational values of $r$ following the approach in \cite{heaukulani2019black} using a simple rejection sampler.
    \item Sampling from the Dirichlet multinomial distribution with \Cref{alg:dirichlet}, which uses a version of the P\'{o}lya urn process \cite{discrete-multivariable-distributions} modified to handle rational fractions of balls. We sparsely encode the output to avoid storing zero entries, as MSDLap sampling only requires summing non-zero random variables.
\end{enumerate}

\begin{algorithm}[ht]
    \caption{Negative binomial sampling optimized for $p > \frac{1}{2}$}
    \label{alg:nb-sample-geo}
\begin{flushleft}
\textbf{Input}: $r \in \mathbb{Q}_{>0}$, $p = e^{-\gamma} = 1- e^{-\eps}$ for $\gamma \in \mathbb{Q}_{>0}$\\
\textbf{Output}: A sample from $\NB(r, p)$
\begin{algorithmic}[1] %[1] enables line numbers
\Loop \Comment{Rejection sampler proposing $\NB(\lceil r \rceil, p)$}
    \State Sample $w \gets \Call{SampleIntegerNB}{\lceil r \rceil, p}$
    \State $A_w \gets (r)_w / (\lceil r \rceil)_w$
    \State Sample $accept \gets \Ber(A_w)$
    \If{$accept$} \textbf{return} $w$ \EndIf
\EndLoop
\Procedure{SampleIntegerNB}{$r$, $p=e^{-\gamma}$}
\State $failures \gets 0$
\State $successes \gets 0$
\Loop
    \State Sample $s \gets \Geo(1-e^{-\gamma})$ \Comment{Use \cite{canonne2020discrete}}
    \State $successes \gets successes + s$
    \If{$successes \ge r$}
        \State \textbf{return } $failures$
    \EndIf
    $failures \gets failures + 1$
\EndLoop
\EndProcedure
\end{algorithmic}
\end{flushleft}

\end{algorithm}
\begin{algorithm}[ht]
    \caption{Sparse Dirichlet multinomial sampling}
    \label{alg:dirichlet}
\begin{flushleft}
\textbf{Input}: $n \in \mathbb{N}$, $k \in \mathbb{N}$, $\alpha = a/b \in \mathbb{Q}_{>0}$\\
\textbf{Output}: A sample from $\DirM(n, \{\alpha,\dots,\alpha\})$, encoded as a sparse map from variate index to count, with zero variates removed. 
\begin{algorithmic}[1] %[1] enables line numbers
\State $picked \gets []$ \Comment{Vector data structure}
\State $initialsize \gets k \cdot a$
\For{$i$ from $0$ to $n-1$}
    \State $urnsize \gets initialsize + i \cdot b$
    \State Sample $U \in \{1, \cdots, urnsize\}$ uniformly at random
    \If{$U < initialsize$}
        \State $idx \gets \lceil U /  a \rceil$
        \State Append $idx$ to $picked$
    \Else
        \State $idx \gets \lceil (U - initialsize) / b \rceil$
        \State Append $picked[idx]$ to $picked$
    \EndIf
\EndFor
\State $counter \gets \{\}$ \Comment{Map data structure}
\For{$p$ in $picked$}
    \If{$p$ in $counter$}
        \State $counter[p] \gets counter[p] + 1$
    \Else
        \State $counter[p] \gets 1$
    \EndIf
\EndFor
\State \textbf{return} $counter$
\end{algorithmic}
\end{flushleft}
\end{algorithm}

\begin{algorithm}[ht]
    \caption{Sparse multivariate i.i.d. $\NB$ sampling}
    \label{alg:multi-nb-sample}
\begin{flushleft}
\textbf{Input}: $k, r \in \mathbb{N}$, $p = e^{-\gamma} = 1-e^{-\eps}$ for $\gamma \in \mathbb{Q}_{>0}$\\
\textbf{Output}: Non-zero samples ${X_1, \dots, X_k}$ where $X_i \sim NB(r, p)$
\begin{algorithmic}[1] %[1] enables line numbers
\State Sample $T \gets \NB(k \cdot r, p)$ \Comment{Use \Cref{alg:nb-sample-geo}}
\State Sample $counter \gets \DirM(T, k, r)$ \Comment{Use \Cref{alg:dirichlet}}
\State \textbf{return} $counter$
\end{algorithmic}
\end{flushleft}
\end{algorithm}

\begin{proposition}\label{prop:nb-sample-from-geo}
%For input $r, \gamma \in \mathbb{Q}_{>0}$ and $p = e^{-\gamma}$, the procedure described in \Cref{alg:nb-sample-geo} returns one sample from $NB(r, p)$ and completes in $O\left(\frac{1}{p} + \frac{r (1-p)}{p^2}\right)$ arithmetic operations in expectation.
For input $r, \gamma \in \mathbb{Q}_{>0}$ and $p = e^{-\gamma}$, the procedure described in \Cref{alg:nb-sample-geo} returns one sample from $NB(r, 1-e^{-\eps})$ and completes in $\widetilde{O}\left(1 + r \cdot e^{-\eps}\right)$ arithmetic operations in expectation, where $\eps = \log \frac{1}{1-e^{-\gamma}}$ and $\widetilde{O}$ hides polynomial factors in $\eps$.

\end{proposition}
\begin{proof}
    We begin by analyzing the case for integer $r$, which amounts to analyzing $\Call{SampleIntegerNB}{}$.
    The approach we use invokes $\Geo(1-p)$ to count
    the \emph{successes} before the first \emph{failure} in
    a sequence of Bernoulli trials. In this way we
    can "batch" runs of successes together.
    
    Let $K$ be the random variable describing the output of the algorithm,
    and $X_i$ be the $i$th $\Geo(1-p)$ random variable.
    Denote $Z \sim \sum_{i=1}^k X_i \sim \NB(k, 1-p)$. 
    First note that $\pr{K = 0} = \pr{X_1 \ge r} = p^r = f_{\NB(r, p)}(0)$. Subsequently for $k \ge 1$,
    \begin{align*}
    \pr{K = k} &= \pr{Z < r \le Z + X_{k+1} }\\
    &= \sum_{z=0}^{r-1} \sum_{x=r-z}^\infty f_{\NB(k, 1-p)}(z) f_{\Geo(1-p)}(x)\\
    &= \sum_{z=0}^{r-1} \sum_{x=r-z}^\infty 
        (1-p)^k p^z \binom{k+z-1}{z} \cdot (1-p)p^x\\
    &= (1-p)^{k+1} \sum_{z=0}^{r-1} \binom{k + z - 1}{k - 1} \frac{p^{r}}{1-p}\\ 
    \text{\Cref{lem:pascal}} &= (1-p)^k p^r \binom{k+r-1}{k}
    = f_{\NB(r, p)}(k).
    \end{align*}
    Thus, the output $K$ follows the $\NB(r, p)$ distribution as desired.

    The expected number of iterations of the loop is exactly $1 + \E[\NB(r, p)]$, since we have one unconditional iteration, and one per increment of $failures$. Each iteration of the loop takes expected $\widetilde{O}(1)$ time, as the geometric sampler requires arithmetic operations polynomial in the bit complexity of $\gamma$ (which is $O(\eps)$).

    Next we analyze the outer loop which handles rational values of $r$ using the accept-reject approach in \cite{heaukulani2019black}, ensuring that in each iteration, for a proposed sample $W$ and acceptance event $A$:
    \begin{align*}
    \pr{W = w \wedge A} &= \pr{A | W=w}\pr{W=w}\\
        &= \frac{(r)_w}{(\lceil r \rceil)_w} f_{\NB(\lceil r \rceil, p)}(w)
        = p^{r - \lceil r \rceil} \cdot f_{\NB(r,p)}(w).
    \end{align*}
    
    This implies that the output follows $\NB(r, p)$ distribution as desired, and that $\pr{A} = p^{r - \lceil r \rceil}$. The latter in turn implies that the number of trials follows a geometric distribution with success probability $p^{\lceil r \rceil - r}$. 
    %\begin{align*}
    %\mathop{\mathbb{E}}_{w \sim \NB(\lceil r \rceil, p)}\left[\frac{(r)_w}{(\lceil r \rceil)_w}\right]
    %&= \mathop{\mathbb{E}}_{w \sim \NB(r, p)}\left[\frac{1}{C}\right]
    %= \frac{1}{C}
    %\end{align*}
    Therefore the expected number of trials is $O(1/p^{\lceil r \rceil - r}) = O(1/p)$,
    and the result follows as $\widetilde{O}\left(\frac{1}{1-e^{-\eps}} \left(1 + \E[\NB(r, 1-e^{-\eps}]\right)\right) = \widetilde{O}(1+r\cdot e^{-\eps})$.
\end{proof}

\begin{proposition}\label{prop:dirichlet-sample}
For input $n, k \in \mathbb{N}$, and $\alpha \in \mathbb{Q}$, the procedure described in \Cref{alg:dirichlet} returns one sample from $\DirM(n, \{\alpha, \cdots, \alpha\})$ and requires $O(n)$ arithmetic operations in expectation.
\end{proposition}
\begin{proof}

We begin the proof by mapping \Cref{alg:dirichlet} to the P\'{o}lya
urn model. The algorithm proceeds as follows:
\begin{itemize}
    \item To initialize the urn contents, $a$ balls are added to the urn for each of the $k$ colors.
    \item The algorithm proceeds over $n$ steps, where at each step a ball is chosen uniformly at random from the urn. When a ball is selected, it is replaced, along with $b$ other balls of that color.
    \item For each color, the algorithm returns the count of how many times that color was chosen.
\end{itemize}

\Cref{alg:dirichlet} implements this in a straightforward manner where $idx$ maps to the $idx$th color $c_{idx}$, and we maintain the list of colors selected in $picked$. We note that the return value is \emph{sparsely encoded} to avoid storing zero entries, or colors that have never been picked.
For purposes of the proof of correctness we assume an unrolled output of $counter$ equal to $X =\{x_1, \cdots, x_k\}$, where each $x_i$ counts the number of times the $c_i$ color was picked. This can be obtained with a simple post-processing step.

Let $S_{m}$ be the color of the ball chosen at iteration $m$.
We evaluate the probability its color is $c$ given the
previous draws from the urn.
\begin{align*}
&\pr{S_{m} = c | S_{1} = s_1 , \dots, S_{m-1} = s_{m-1}}\\
&= \frac{a + \sum_{i=1}^{m-1} b \cdot s_i}{k \cdot a + b(m -1)} &\text{let $z = \sum_{i=1}^{m-1} \mathbbm{1}(s_i = c)$}\\
&= \frac{a + b \cdot z}{k\cdot a + b(m -1)}.\\
\end{align*}

Note that $S_{m}$ only depends about the current state of the urn, not the order in which balls are picked. Similarly,
note that the denominator of the fraction is independent of
any information about $z$. From this we can show that the probability of seeing any one \emph{particular sequence} $S = \{s_1, \cdots, s_n\}$ which has color counts $X = \{x_1, \dots, x_k\}$ is

\begin{align*}
\pr{S_1 = s_1, \Compactcdots, S_n = s_n} &= \left(\prod_{i=1}^n 
             \frac{1}{k\cdot a + b(i - 1)}
             \right)
             \prod_{j=1}^k
             \prod_{t=1}^{x_i} a + b(t - 1)\\
   &= \frac{\prod_{j=1}^k b^{x_i} (a/b)_{x_i}}{b^n (k\cdot a/b)_n}\\
    &= \frac{\Gamma(k\cdot \alpha)}{\Gamma(n + k\cdot \alpha)}
        \prod_{j=1}^k \frac{\Gamma(\alpha + x_i)}{\Gamma(\alpha)}.
\end{align*}
Furthermore, this quantity is order agnostic. Any
reordering of the draws has the same probability.
The correctness result follows from noting that there are
$$
\binom{n}{x_1, x_2,\dots, x_k} = \frac{n!}{x_1! \cdots x_k!} = \frac{\Gamma(n + 1)}{\prod_{i=1}^k \Gamma(x_i+1)}
$$ possible orderings of the draws, so
\begin{align*}
\pr{X = x} &= \frac{\Gamma(k\cdot \alpha)\Gamma(n+1)}{\Gamma(n + k\cdot \alpha)}
        \prod_{j=1}^k \frac{\Gamma(\alpha + x_i)}{\Gamma(x_i + 1)\Gamma(\alpha)}
= f_{\DirM(n, \{\alpha, \dots, \alpha\})}(x).
\end{align*}

For the run time analysis, both loops in the algorithm iterate exactly $n$ times,
and each require only constant arithmetic operations
in expectation, assuming $O(1)$ map and vector operations.
\end{proof}

\Cref{lem:nb-conditioned-sum-is-dirm}, \Cref{prop:nb-sample-from-geo}, and \Cref{prop:dirichlet-sample} immediately show \Cref{thm:multi-nb-sample}.

%\begin{theorem}\label{thm:multi-nb-sample}
%For input $k, r \in \mathbb{N}$, and $p = e^{-\gamma}$ for $\gamma > 0 \in \mathbb{Q}$, the procedure described in \Cref{alg:multi-nb-sample} returns non-zero samples from $k$ i.i.d. samples of $\NB(r, p)$ and completes in $O\left(\frac{1}{p} + \frac{k \cdot r (1-p)}{p^2}\right)$ steps in expectation.
%For input $k, r \in \mathbb{N}$, and $p = e^{-\gamma}$ for $\gamma > 0 \in \mathbb{Q}$, the procedure described in \Cref{alg:multi-nb-sample} returns non-zero samples from $k$ i.i.d. samples of $\NB(r, 1-e^{-\eps})$ and completes in $\widetilde{O}\left(1 + k \cdot r \cdot e^{-\eps}\right)$ steps in expectation, where $\eps = \log \frac{1}{1-e^{-\gamma}}$ and $\widetilde{O}$ hides polynomial factors in $\eps$.

%\end{theorem}
%\begin{proof}
%This follows directly from \Cref{lem:nb-conditioned-sum-is-%dirm}, \Cref{prop:nb-sample-from-geo}, and \Cref{prop:dirichlet-%sample}.
%\end{proof}

%\Cref{thm:multi-nb-sample} immediately implies %\Cref{cor:constant-time-sampler} stated earlier.

%\begin{remark}
%\Cref{cor:constant-time-sampler} suggests that our sampler may be relevant even in the non-distributed central model, i.e. a single party computing a large number of discrete Laplace random variables. In this case, $r = 1$ and we can compute $O(e^{a})$ $\DLap(a)$ random variables in constant expected time.
%\end{remark}

\newcommand{\cY}{\mathcal{Y}}
\newcommand{\cO}{\mathcal{O}}

\section{Order optimal MSE in the multi-message shuffle model}\label{sec:shuffle}

In this section we use the mechanisms from \Cref{sec:discrete} to extend multi-message protocols in the shuffle model of differential privacy.

First, we recall the definition of the shuffle model~\cite{CheuSUZZ19,ErlingssonFMRTT19}. An $m$-message protocol in the shuffle model consists of a \emph{randomizer} $\cR: \cX \to \cY^m$ where $\cY$ denote the set of all possible messages and an \emph{analyzer} $\cA: \cY^{n m} \to \cO$ where $\cO$ denotes the output domain. In the shuffle model, the analyst does not see the output of each randomizer directly but only the randomly shuffled messages. We write $\cS_{\cR}(x_1, \dots, x_n)$ to denote the output of randomly shuffling $nm$ (random) messages produced by $\cR(x_1), \dots, \cR(x_n)$ for $x_1, \dots, x_n \in \cX$. The shuffle model required that these shuffled messages have to satisfy DP, as stated more formally below.

\begin{definition}[\cite{CheuSUZZ19,ErlingssonFMRTT19}]
An $m$-message protocol $(\cR, \cA)$ is $(\eps, \delta)$-shuffle-DP if, for every $x, x' \in \cX^n$ differing in a single entry, 
$
\Pr[\cS_{\cR}(x) \in S] \leq e^{\eps} \cdot \Pr[\cS_{\cR}(x') \in S] + \delta
$ 
for all $S \subseteq \cY^{nm}$.
\end{definition}

In the real summation problem, each $x_i$ is a real number in $[0, 1]$ and the goal is to compute $\sum_{i \in [n]} x_i$. Our main result of this section can be stated as follows:
\begin{theorem} \label{thm:shuffle-dp-main}
For every $\eps \geq 2$ and $\delta \in (0, 1/n)$, there exists an $(\eps, \delta)$-shuffle-DP, $O\left(\frac{\eps + \log(1/\delta)}{\log(n)}\right)$-message protocol for real summation with MSE $O\left(e^{-2\eps/3}\right)$ such that each message is $O\left(\eps + \log n\right)$ bits long.
\end{theorem}

Prior to this work, the best known protocol has MSE that scales as $O(1/\eps^2)$ (even for large $\eps$)~\cite{balle2020private,GhaziMPV20,almost-central-shuffle-ghazi21a}\footnote{Pagh and Stausholm~\cite[Corollary 23]{pagh2022infinitely} claims that their result implies a protocol with absolute error $\frac{1}{e^{\Omega(\eps)} -1}$, but, to the best of our knowledge, this is not the case. See the discussion at the end of this section for more detail.} and thus our result provides a significant improvement in the large $\eps$ regime.

To prove this result, it would be convenient to define the $\Z_q$-summation problem as follows. Here the input $x_1, \dots, x_n$ belongs to $\Z_q$ and the goal is to compute $\sum_{i \in [n]} \Z_q$. Similar to before, an $m$-message protocol consists of a randomizer $\cR$ and an analyzer $\cA$. We say that the protocol is \emph{exact} if the analyzer always output the correct answer. Furthermore, we say that the protocol is \emph{$\sigma$-secure} (for some parameter $\sigma > 0$) if, for all $x, x' \in \Z_q^n$ that results in the same output (i.e. $\sum_{i \in [n]} x_i = \sum_{i \in [n]} x'_i$), we have $\TV(\cS_{\cR}(x), \cS_{\cR}(x')) \leq 2^{-\sigma}$. Building on the "split and mix" protocol of \cite{ishai2006cryptography}, Balle at al.~\cite{balle2020private} and Ghazi et al.~\cite{GhaziMPV20} gave a secure $\Z_q$-summation protocol with the following guarantee.

\begin{theorem}[\cite{balle2020private,GhaziMPV20}] \label{thm:split-and-mix}
For any $\sigma \in (0, 1/2)$ and $q \in \N$, there is an $\sigma$-secure $O\left(1 + \frac{\sigma}{\log n}\right)$-message protocol for $\Z_q$-summation in the shuffle model where each message belongs to $\Z_q$.
\end{theorem}

Balle et al.~\cite{balle2020private} presented an elegant method to translate a secure $\Z_q$-summation protocol to a shuffle-DP real summation protocol. Below, we provide a slight generalization and improvement\footnote{Originally, their analysis has an additional error term due to ``overflow / underflow''. We observe that this is in fact unnecessary, which reduces the claimed error and also simplifies the analysis.} of their result, which will eventually allow us to combine~\Cref{thm:split-and-mix} and our infinitely divisible noise to achieve~\Cref{thm:shuffle-dp-main}.

%In particular, we show a differentially private extension of the "split and mix" protocol of \cite{ishai2006cryptography} to show improved MSE in the high $\eps$ regime.
%In this protocol, each party holds a value which they locally noise and break into $m$ separate additive secret shares. The analyzer sees a shuffled view of the shares of all parties, and simply sums them up. Our approach will show a version of \cite[Lemma 5.3]{balle2020private} that is independent of the noise distribution, rather than using the discrete Laplace.

\begin{lemma}[generalization of \cite{balle2020private} Lemma 5.2]\label{lem:balle-lem53}
Suppose that, for some $\eps > 0, \Delta \in \N$, there is a zero-mean discrete infinitely divisible distribution $\cD$ such that the $\cD$-noise addition mechanism is $\eps$-DP for sensitivity $\Delta$.
Furthermore, suppose that, for some $q, n \in \N$ with $q \geq 2\Delta n$ and $\sigma \in (0, 1/2)$, there exists an $n$-party $\sigma$-secure $m$-message exact $\Z_q$-summation protocol in the shuffle model. Then, there is an $m$-message $(\eps, (1+e^\eps)2^{-\sigma})$-shuffle-DP protocol for real summation with MSE
$\frac{\Var(\cD)}{\Delta^2} + \frac{n}{4\Delta^2}$.

Moreover, the message length is the same as in the $\Z_q$-summation protocol.
\end{lemma}

\begin{proof}
Let $\Pi = (\mathcal{R}_\Pi, \mathcal{A}_\Pi)$ be the $\sigma$-secure exact $\Z_q$-summation protocol. Our protocol $\mathcal{P} = (\mathcal{R}_\Pi \circ \mathcal{R}, \mathcal{A} \circ \mathcal{A}_\Pi)$ where $\cR, \cA$ are defined as follows\footnote{Note that the final protocol $\cP$ is done by composing the randomizer $\cR$ / analyzer $\cA$ with those of $\Pi$.}:

\begin{enumerate}
    \item $\mathcal{R}: [0, 1] \to \mathbb{Z}_{q}$ on input $x_i$ works as follows:
    \begin{itemize}
    \item First, it computes a randomized encoding $y_i$ by setting
    \begin{align*}
    y_i = 
    \begin{cases}
    1 + \lfloor \Delta x_i \rfloor & \text{ w.p. } \Delta x_i - \lfloor \Delta x_i \rfloor \\
    \lfloor \Delta x_i \rfloor & \text{ w.p. } 1 - (\Delta x_i - \lfloor \Delta x_i \rfloor)
    \end{cases}
    \end{align*}
    \item Then, it samples $Z_i \sim \cD_{/n}$
    \item Finally, it outputs $(y_i + Z_i) \mod q$.
    \end{itemize}
    \item $\mathcal{A}$ decodes the result $r$ by returning 
    \begin{align*}
    r' =
    \begin{cases}
    r / \Delta &\text{ if } 0 \leq r \leq n \Delta, \\
    n &\text{ if } n \Delta + 1 \leq r \leq 2n\Delta, \\
    0 &\text{ otherwise.}
    \end{cases}
    \end{align*}
\end{enumerate}

\emph{(Privacy)} The proof of privacy proceeds identically to \cite[Lemma 5.2]{balle2020private} and is omitted here.

\emph{(Accuracy)}
Since $\Pi$ is an exact $\Z_q$-summation protocol, its output $r$ is exactly equal to $\tu \mod q$ where $\tu = \sum_{i \in [n]} (y_i + z_i) = Z + \sum_{i \in [n]} y_i$ where $Z = z_1 + \cdots + z_n$ is distributed as $\cD$.

Let $u = \sum_{i \in [n]} x_i$ be the true (unnoised) sum.
The first step in our accuracy analysis is a claim\footnote{This claim is indeed our improvement over the error analysis of \cite[Lemma 5.2]{balle2020private}.} that, regardless of the value of $u$, we have $|r' - u| \leq |\tu/\Delta - u|$. To see that this is true, let us consider the following cases:
\begin{itemize}
\item Case I: $\tu \notin (-n\Delta, 2n\Delta)$. In this case, $|\tu/\Delta - u| \geq n \geq |r' - u|$.
\item Case II: $\tu \in [0, n\Delta]$. In this case, we have $r' = \tu/\Delta$ and thus the inequality holds as an equality.
\item Case III: $\tu \in (-n\Delta, 0)$. In this case, we set $r' = 0$. Thus, $|\tu/\Delta - u| = u - \tu/\Delta \geq u - r' = |r' - u|$
\item Case IV: $\tu \in (n\Delta, 2n\Delta)$. In this case, we set $r' = n$. Thus, $|\tu/\Delta - u| = \tu/\Delta - u \geq r' - u = |r' - u|$.
\end{itemize}
Thus, in all cases, we have $|r' - u| \leq |\tu/\Delta - u|$.

We can then bound the MSE of the protocol as follows:
\begin{align*}
\E\left[\left(r' - u\right)^2\right]
&\leq \E\left[\left(\tu/\Delta - u\right)^2\right] \\
&= \E\left[\left(Z/\Delta + \sum_{i \in [n]} (y_i/\Delta - x_i)\right)^2\right] \\
&= \E\left[\left(Z/\Delta\right)^2\right] + \sum_{i \in [n]} \E\left[\left(y_i/\Delta - x_i\right)^2\right] \\
&\leq \frac{\Var(\cD)}{\Delta^2} + \frac{n}{4\Delta^2}. \qedhere
\end{align*}
\end{proof}

We can now prove \Cref{thm:shuffle-dp-main} by plugging our infinitely divisible noise (\Cref{cor:msdlap-rextremes}) to the above lemma.

\begin{proof}[Proof of \Cref{thm:shuffle-dp-main}]
Let $\Delta = \left\lceil e^{\eps/3} \sqrt{n} \right\rceil, q = 2 n \Delta$ and $\sigma = \log_2\left(\frac{e^{\eps} + 1}{\delta}\right)$. From \Cref{cor:msdlap-rextremes}, there is a zero-mean discrete infinitely divisible distribution $\cD$ such that $\cD$-noise addition mechanism is $\eps$-DP for sensitivity $\Delta$ where $\Var(\cD) \leq O(\Delta^2 \cdot e^{-2\eps/3})$. Furthermore, \Cref{thm:split-and-mix} ensures that there exists an $n$-party $\sigma$-secure $m$-message exact $\Z_q$-summation protocol where $m = O\left(1 + \frac{\sigma}{\log n}\right) = O\left(1+\frac{\eps + \log(1/\delta)}{\log n}\right)$. Thus, \Cref{lem:balle-lem53} implies that there is an $(\eps, \delta)$-shuffle-DP $m$-message protocol for real summation where each message is of length $O(\log q) = O(\eps + \log n)$ bits and MSE 
\begin{align*}
\frac{\Var(\cD)}{\Delta^2} + \frac{n}{4\Delta^2} \leq O(e^{-2\eps/3}). & \qedhere
\end{align*}
\end{proof}

Notice that, in the above proof, we use the noise from \Cref{cor:msdlap-rextremes}, which is a modified version of MSDLap.
We remark that a similar result \emph{cannot} immediately be shown via the GDL (\Cref{thm:opt-var}) or the vanilla MSDLap distributions with $r=0$ (\Cref{thm:multi-scale-dlap}),
as their $\Delta^3$ dependence in MSE leads to unavoidable dependence on $n$ in the final MSE error, due to the randomized rounding. We also stress that the continuous
results in \Cref{sec:integer-to-real} or \cite{pagh2022infinitely} similarly cannot be used here,
as the protocol must round each contribution to a finite group $\mathbb{Z}_q$ prior to summing. Neither the privacy nor the infinite divisibility of the resulting sum distribution is clear in this case.

\section{Conclusion}

This work closes the utility gap for infinitely divisible mechanisms in the high $\eps$ pure DP regime. We find no separation in either the discrete or continuous settings by restricting the private mechanism to infinitely divisible noise addition. The "staircase-like" MSE of both the GDL and the multi-scale discrete Laplace in the low-privacy regime make them a natural replacement for the staircase mechanism in the discrete distributed setting, and we hope the results introduced here can be of practical value. We show one such practical application by extending the \cite{ishai2006cryptography}
"split and mix" protocol under shuffle differential privacy,  resolving the open question posed in \cite{almost-central-shuffle-ghazi21a}. %Extending our results to other multi-message protocols (e.g. the correlated noise protocol in \cite{almost-central-shuffle-ghazi21a}) seems plausible.

Beyond improving utility, we believe the generalized discrete Laplace is of independent interest for distributed private mechanism design due to the fact that it is closed under summation.
This makes it well-suited for cases where "smooth" privacy guarantees are needed for multiple outcomes, or for a single deployed system (like a secure aggregation MPC protocol over $n$ clients e.g. \cite{secagg}) where different people can make \emph{different} assumptions about what the honest fraction of the involved clients are. Additionally, this property is useful in \emph{post-hoc} privacy loss analysis when honesty assumptions are broken in a production system, and where otherwise an analyst must resort to numerical approximation of realized privacy loss.

In the continuous distributed setting, our continuous transformations of all of the mechanisms in \Cref{sec:discrete} outperform the existing Arete mechanism, as we plot in \Cref{fig:cont}. We also note a strong resemblance between the Arete distribution and the GDL mechanism, as the negative binomial distribution converges to the gamma under certain conditions. We explore this convergence formally in \Cref{app:arete-conv}.

Finally, we hope our optimized MSDLap sampler \Cref{alg:multi-nb-sample} or its constituent parts can be useful in any context where sparse, exact multivariate $\NB$ random generation is needed, or even where a general $\NB(r,p)$ sampler needs to be sublinear in $r$ (\Cref{alg:nb-sample-geo}). For multivariate sampling when $p$ is very close to 1, our approach should
be a large improvement over standard methods. We note that \Cref{alg:dirichlet} can be extended in a straightforward way (due to \Cref{lem:nb-conditioned-sum-is-dirm}) to support $\NB$ variates with different $r$ values. The only change to the algorithm is initializing the urn with varying numbers of balls.

\subsection{Open questions}
In preparing this paper, we studied the Arete mechanism \cite{pagh2022infinitely} extensively and are convinced it can also achieve the order optimal MSE of $O\left(\Delta^2 e^{-2\eps/3}\right)$, but the formal proof eluded us. We present the following conjecture\footnote{Even if this conjecture were proven, the constant factors for the Arete's MSE still underperform the continuous multi-scale discrete Laplace.%, though the Arete distribution is simpler to sample from.
} as an open question:

\begin{conjecture}\label{conj:arete}
    Let $\Delta \ge 1$, $\eps > 10$. The $\Arete(\alpha, \theta, \lambda)$-noise addition mechanism with $\alpha = e^{-2 \eps / 3}$, $\theta = (1+t) \frac{1}{\log 2}$, and $\lambda = (1+t) e^{-\eps / 3}$ where $t = o(1)$
    is $\eps$-DP for sensitivity $\Delta$ and has MSE $O\left(\Delta^2 e^{-2\eps/3}\right)$.
\end{conjecture}

Our work also raises the following questions:

\begin{itemize}
\item 
How close to the optimal MSE \emph{including constants} can be achieved in the for continuous and infinite divisible mechanisms? While we match the constants of the discrete staircase for large $\eps$, there is still a sizable gap between our results and the continuous staircase mechanism in constant factors.
%\item Does there exist an exact $\NB(r, p)$ sampler which requires fewer than $O(\min(r, \E[\NB(r, p)])$ arithmetic operations in expectation?
% I decided to remove the above question because it is not
% of critical importance and indeed would not even change
% the running time of our MSDLap sampler.
\item Can the number of \emph{messages} required for our shuffle-DP protocol in \Cref{thm:shuffle-dp-main} be further improved? Recall that we use the approach of Balle et al. \cite{balle2020private} to translate a secure summation protocol to a shuffle-DP one. In fact, there is a more direct approach by Ghazi et al.~\cite{almost-central-shuffle-ghazi21a} that achieves a smaller number of messages. Since their proof also uses infinitely divisible noises, it is plausible that our noise distribution can be used there. However, their proof is considerably more involved compared to \cite{balle2020private} and does not use the noise distributions in a black-box manner. Therefore, we leave it as a future research direction.

\end{itemize}


\section*{Acknowledgment}
Thanks to Peter Kairouz for useful discussions and pointing out the prior work, especially \cite{bagdasaryan2022heatmaps}. This work was supported by Google.

\bibliographystyle{alpha}
\bibliography{refs}

\appendix

\section{Appendix}

\subsection{Deferred Proof of \Cref{lem:logconvex-ratio}}
\label{app:logconvex-ratio-proof}

Before we prove \Cref{lem:logconvex-ratio}, it will be convenient to state the following simple lemma.

\begin{lemma}\label{lem:logconvex-app}
    Let $f(x)$ be log convex on the interval $[a, b]$. Then for any $x \in [a,b]$ and $\Delta \in \mathbb{R}^+$ such that $x+\Delta \in [a,b]$:
    $$\frac{f(0)}{f(\Delta)} \ge \frac{f(x)}{f(x + \Delta)}$$
\end{lemma}

\begin{proof}
From log-convexity, we have
\begin{align*}
f(0)^{\frac{\Delta}{x + \Delta}} f(x + \Delta)^{\frac{x}{x + \Delta}} &\geq f(x), &\text{ and }  \\
f(0)^{\frac{x}{x + \Delta}} f(x + \Delta)^{\frac{\Delta}{x + \Delta}} &\geq f(\Delta).
\end{align*}
Multiplying the two yields the claimed inequality.
\end{proof}

We are now ready to prove \Cref{lem:logconvex-ratio}.

\begin{proof}[Proof of \Cref{lem:logconvex-ratio}]
    Assume w.l.o.g. that $x \le x'$. There are three cases to consider.
    \begin{enumerate}
        \item $x < 0$ and $x' < 0$. We have $\frac{f(x)}{f(x')} \le \frac{f(-x)}{f(-x')}$.
        \item $x < 0$ and $'x \ge 0$. We have $\frac{f(x)}{f(x')} \le \frac{f(0)}{f(x')}$
        \item $x \ge 0$ and $x' \ge 0$.
    \end{enumerate}
    So it suffices to handle case (3). By $f$ decreasing on $[0, \infty)$, we have $\frac{f(x)}{f(x')} \le \frac{f(x)}{f(x+\Delta)}$. Applying \Cref{lem:logconvex-app} concludes the proof.
\end{proof}

\subsection{Alternative simplified GDL privacy bound}\label{sec:alt-gdl-privacy}

In this section we will outline a tighter version of
\Cref{cor:hyper-ratio}, which only well-approximates the privacy
loss in the small $\beta$ regime. This bound well-approximates
the privacy loss in all $\beta$ regimes, at the cost of some added complexity in the expression.

\begin{corollary}[to \Cref{thm:main-epsilon}]
For any $\Delta \in N, a > 0, \beta \in (0, 1)$, the GDL$(\beta, a)$-noise addition mechanism is $\eps$-DP for sensitivity $\Delta$ where $\eps \le a \Delta + (1-\beta)\log(\beta + \Delta) + \log(\Gamma(\beta))$
\end{corollary}
\begin{proof}
    The proof is identical to \Cref{cor:hyper-ratio}, except in the last step where we use the following Wendel's double inequality \cite[eq. 2.6]{qi2010bounds}:

\begin{align*}
\log\left(\frac{\Gamma(\Delta + 1)\Gamma(\beta)}{\Gamma(\beta + \Delta)}\right) 
&\le \log\left(\Gamma(\beta)(\Delta+\beta)^{1-\beta}\right)\\
&= (1-\beta)\log(\beta + \Delta) + \log(\Gamma(\beta)). &\qedhere
\end{align*}
\end{proof}
\subsection{Analytical variance of the discrete staircase distribution}\label{sec:dstair-var}

In this section we derive the analytical variance of the discrete staircase
using the Mathematica software \cite{Mathematica}, for the purposes of generating
\Cref{fig:discrete}.

\begin{figure}[htp]
  \centering
  \includegraphics[width=\linewidth]{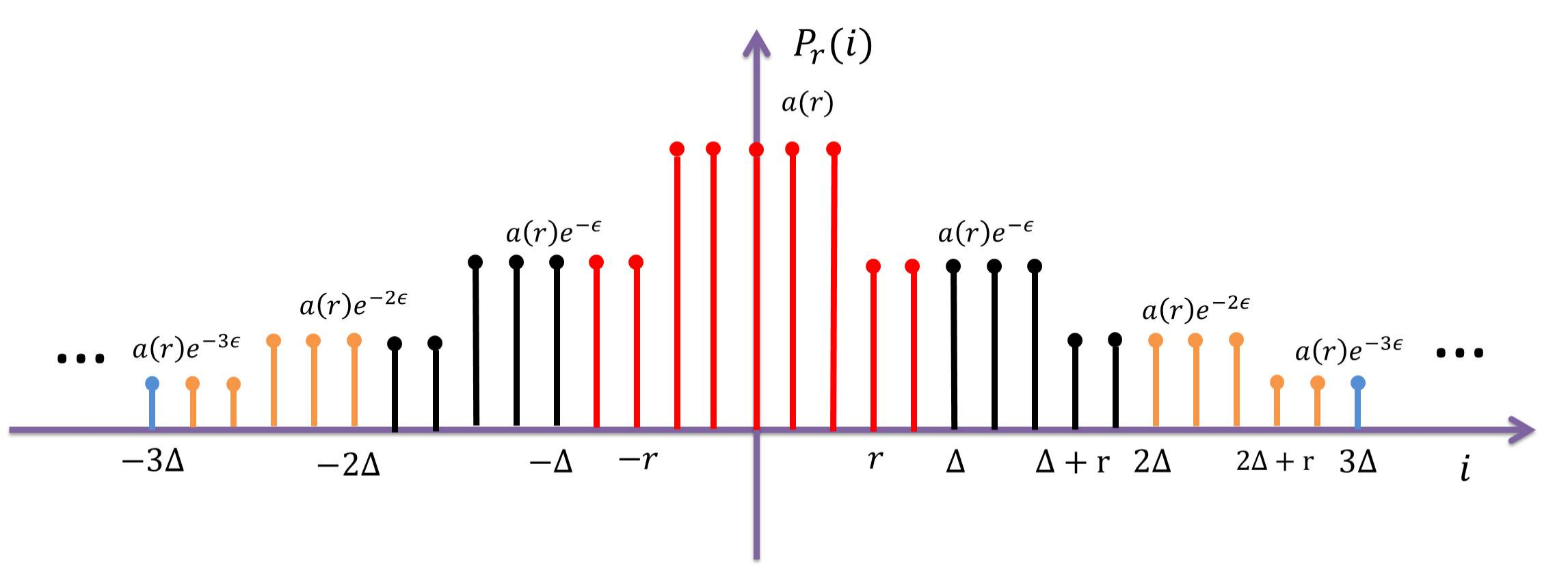}
  \caption{The Discrete Staircase PMF from \cite[Figure 5]{geng2014optimal}}
  \label{fig:dstair}
  \Description{A visual of the PMF in \Cref{def:dstair}. It shows a symmetric PMF, with a central column of width $2 r - 1$ with probability mass $a(r)$ at each point. Subsequent columns are of width $\Delta$ and their probability mass is $e^{-\eps}$ times smaller than the column nearer to the center.}
\end{figure}

\begin{definition}[\cite{geng2014optimal}]\label{def:dstair}
The discrete staircase distribution with parameters $1 \le r \le \Delta$, $\eps > 0$, and $\Delta \in \mathbb{N}$ is defined as
$$
f_{\DStair_{r, \eps, \Delta}}(i) = 
\begin{cases}
a(r) & 0 \le i < r \\
a(r)e^{-\epsilon} & r \le i < \Delta\\
e^{-k \epsilon}f_{\DStair_{r,\eps, \Delta}}(i - k\Delta) & k \Delta \le i < (k + 1) \Delta\\
f_{\DStair_{r,\eps, \Delta}}(-i) & i < 0
\end{cases}
$$

where  $a(r) = \frac{1-b}{2r + 2b(\Delta - r) - (1-b)}$, $b = e^{-\epsilon}$, and $k \in \mathbb{N}$.
\end{definition}

\begin{lemma}
Let $z = e^\eps - 1$, then
\begin{align*}
\Var[\DStair_{r, \eps, \Delta}] = \frac{x_1 + x_2 + x_3}{3z^2 \left(1-2r +e^\eps (2r -1) + 2\Delta\right)}
\end{align*}
Where
\begin{itemize}
\item $x_1 = 2r^3 z^3 - 3r^2 z^2(z - 2\Delta)$
\item $x_2 = r z(1 + e^{2\eps} + 6\Delta(1+\Delta) + e^\eps (6\Delta(\Delta-1) - 2))$
\item $x_3 = 2e^\eps \Delta (-1 + 4\Delta^2 + \cosh(\eps) + 2\Delta^2 \cosh(\eps) - 3\Delta \sinh(\eps))$
\end{itemize}
\end{lemma}
\begin{proof}
Let $X \sim \DStair(r, \eps, \Delta)$. We will first compute the variance from just the central "stair".
$$
C = \sum_{x=-r+1}^{r-1} x^2 a(r) = \frac{1}{3}r(r-1)(2r - 1)a(r)
$$

Finally we compute the variance from the full support of the distribution.

\begin{align*}
\E[X^2] &= \sum_{x=-\infty}^\infty x^2 f_{\DStair_{r, \eps, \Delta}}(x)\\
&= C + 2\sum_{k=1}^\infty \sum_{i=1}^\Delta ((k-1) \Delta + r + i - 1)^2 f_{\DStair_{r, \eps, \Delta}}(k\Delta -r + i)\\
&= C + 2\sum_{k=1}^\infty \sum_{i=1}^\Delta ((k-1) \Delta + r + i - 1)^2 a(r)e^{-k\epsilon}\\
\end{align*}

The result follows from symbolic simplification in Mathematica, and replacing $e^\eps - 1$ terms with $z$ for ease of presentation.
\end{proof}

We also derive the following concrete bound for large $\eps$, which facilitates a comparison with our MSDLap noise.

\begin{observation} \label{obs:dstair-var-large}
For any $\Delta$ and $\eps \geq \log\left(\frac{\Delta (\Delta + 1)(2\Delta+1)}{2}\right)$, the variance of the discrete staircase
    distribution (for any value of $r$) is at least $\frac{1}{e^{\eps} + (2\Delta - 1)} \cdot \frac{\Delta (\Delta + 1)(2\Delta+1)}{3}$
\end{observation}

\begin{proof}
First, notice that, if $r \ne 1$, then $f_{\DStair_{r, \eps, \Delta}(0)} \leq \frac{1}{3}$. Since the noise is zero-mean, the variance is thus at least $\frac{2}{3}$ which is at least $\frac{1}{e^{\eps} + (2\Delta - 1)} \cdot \frac{\Delta (\Delta + 1)(2\Delta+1)}{3}$ by our condition on $\eps$.

Next, consider the case $r = 1$. In this case, the variance is exactly
\begin{align*}
&\sum_{i=-\infty}^{\infty}  i^2 \cdot f_{\DStair_{1, \eps, \Delta}}(i) \\
&= 2 \sum_{i=1}^{\infty}  i^2 \cdot f_{\DStair_{1, \eps, \Delta}}(i) \\
&= 2 \sum_{j=1}^{\Delta} \sum_{\ell=0}^{\infty} (j + \ell \Delta)^2  
 \cdot f_{\DStair_{1, \eps, \Delta}}(j+\ell \Delta) \\
 &= 2 \sum_{j=1}^{\Delta} \sum_{\ell=0}^{\infty} (j + \ell \Delta)^2  
 \cdot a(1) \cdot e^{-\eps(\ell+1)} \\
 &\geq 2 \sum_{j=1}^{\Delta} \sum_{\ell=0}^{\infty} j^2  
 \cdot a(1) \cdot e^{-\eps(\ell+1)} \\
 &= 2a(1) \left(\sum_{j=1}^{\Delta} j^2\right)\left(\sum_{\ell}^{\infty} e^{-\eps(\ell + 1)}\right) \\
 &= 2 \left(\frac{1 - e^{-\eps}}{1 + (2\Delta - 1)e^{-\eps}}\right) \left(\frac{\Delta (\Delta + 1)(2\Delta+1)}{6}\right) \left(\frac{e^{-\eps}}{1-e^{-\eps}}\right) \\
 &= \frac{1}{e^{\eps} + (2\Delta - 1)} \cdot \frac{\Delta (\Delta + 1)(2\Delta+1)}{3}
\end{align*}
where $a(1)$ is as defined in \Cref{def:dstair}.
\end{proof}

\iffalse
We also show the following observation from symbolic simplification.
\begin{observation}\label{obs:dstair-var-large}
    For a fixed $\Delta$ and arbitrarily large $\eps$, the variance of the discrete staircase
    distribution is minimized by $r = 1$ by $ \frac{1}{3} e^{-\eps} \Delta (\Delta + 1)(2\Delta+1)$
\end{observation}

\begin{proof}
    That $r=1$ minimizes the variance is clear as it is the only parametrization whose central column has width 1 (and thus contributes 0 to the variance). The rest of the distribution's contribution to the variance decreases with $\eps$.
    To show the result, we take the limit with Mathematica.

    $$
    \lim_{\eps \to \infty} e^\eps \Var[\DStair_{1, \eps, \Delta}]
    $$
\end{proof}
\fi
\subsection{Arete convergence}\label{app:arete-conv}
\todo{WIP section}
In this section we show a link between the GDL
distribution and the Arete distribution from \cite{pagh2022infinitely}.

\begin{definition}\label{def:arete}
For $k, \theta, \lambda > 0$, let $\Arete(k, \theta, \lambda)$ denote the distribution of $Z_1 - Z_2 + Z_3$ where $Z_1, Z_2 \overset{\text{i.i.d.}}{\sim} \Gamma(k, \theta)$ and $Z_3 \sim \Lap(\lambda)$.
\end{definition}

The main technical lemma linking the $\GDL$ and $\Arete$ follows from showing how
the negative binomial distribution converges to the gamma distribution.

\begin{lemma}\label{lem:nb-to-gamma}
Let $X \sim \NB(k, 1-e^{-\frac{1}{\theta \Delta_\disc}})$ and $Z_{\Delta_\disc} \sim X / \Delta_\disc$. Then
    $$
    Z_{\Delta_\disc} \xrightarrow{dist} \Gamma(k, \theta)
    $$
as $\Delta_\disc \to \infty$.
\end{lemma}
\begin{proof}
By Lévy's continuity theorem, convergence in distribution
follows from pointwise convergence of the characteristic function. Denote $\varphi(\mathcal{D})$ the characteristic function of the distribution $\mathcal{D}$. We first note the following facts:
\begin{itemize}
\item $\varphi(\NB(r, p)) = \left(\frac{p}{1-e^{i t}(1-p)}\right)^r$
\item $\varphi(\NB(r, p) / \Delta_\disc) = \left(\frac{p}{1-e^{i t / \Delta_\disc}(1-p)}\right)^r$
\item $\varphi(\Gamma(k, \theta)) = (1 - \theta t)^{-k}$
\end{itemize}
Finally,
\begin{align*}
\lim_{\Delta_\disc \to \infty} \varphi(Z) &= \lim_{\Delta_\disc \to \infty} 
    \left(
    \frac{1 - e^{-\frac{1}{\theta \Delta_\disc}}}
        {1-e^{(i t - 1/\theta) / \Delta_\disc}}
    \right)^k\\
\text{L'Hôpital} &= \left(
        \lim_{\Delta_\disc \to \infty}
        \frac{\frac{\mathrm{d}}{\mathrm{d} \Delta_\disc} 1 - e^{-\frac{1}{\theta \Delta_\disc}}}
        {\frac{\mathrm{d}}{\mathrm{d} \Delta_\disc} 1-e^{(i t - 1/\theta) / \Delta_\disc}}
    \right)^k\\
&= \left(
 \lim_{\Delta_\disc \to \infty}
-\frac{e^{-\frac{1}{\theta \Delta_\disc}}}{\theta \Delta_\disc^2} \cdot
   \frac{\Delta_\disc^2}
   {e^{\frac{i t - 1/\theta}{\Delta_\disc}}(i t - 1/\theta)}
   \right)^k\\
&= \left(
 \lim_{\Delta_\disc \to \infty}
\frac{e^{-\frac{1}{\theta \Delta_\disc} - \frac{i t}{\Delta_\disc} +\frac{1}{\theta \Delta_\disc}}}{-\theta (i t - 1/\theta))}
   \right)^k\\
&= \left(
 \lim_{\Delta_\disc \to \infty}
\frac{e^{- \frac{i t}{\Delta_\disc}}}{1 - i t \theta}
   \right)^k\\
&= \left(
        \lim_{\Delta_\disc \to \infty}
        \frac{i e^{\frac{- i t \theta}{\Delta_\disc \theta}}}
             {i + t \theta}
    \right)^k\\
&= \left(
    \frac{i}{i + t \theta}
   \right)^k\\
&= \varphi(\Gamma(k, \theta)). &\qedhere
\end{align*}
\end{proof}

\begin{proposition}\label{prop:arete-conv}
Let $Z_{\Delta_\disc} \sim \Lap(\lambda) + \GDL(k, \frac{1}{\theta \Delta_\disc})/\Delta_\disc$.
Then $Z_{\Delta_\disc} \xrightarrow{dist} \Arete(k, \theta, \lambda)$ as $\Delta_\disc \to \infty$.
\end{proposition}
\begin{proof}
This follows from \Cref{def:gdl}, \Cref{def:arete}, and \Cref{lem:nb-to-gamma}.
\end{proof}

\begin{remark}
 The convergence result in \Cref{prop:arete-conv} shows that the Arete mechanism is quite similar to the GDL mechanism transformed with real support via the approach in \Cref{thm:discrete-to-cont}. The primary difference is how large $\Delta_\disc$ gets. Using \Cref{thm:discrete-to-cont}, we only set $\Delta_\disc = O(e^{\eps/3})$, allowing the resulting (discrete) distribution to have "holes" in its support. The purpose of the Laplace noise in that case is to smooth out the holes and ensure support on $\mathbb{R}$. On the other hand the Arete (via \Cref{prop:arete-conv}) requires $\Delta_\disc \to \infty$ (with no holes in its support), and the purpose of the Laplace noise is to smooth out the resulting \emph{singularity} at 0 for sufficiently small values of $k$.\footnote{See \cite[Page 3]{pagh2022infinitely} for further discussion on this point.} As such, the proof technique for \Cref{thm:discrete-to-cont} cannot be immediately used to help prove (or disprove) \Cref{conj:arete}. 
\end{remark}
\newcommand{\sdiff}{S_{\mathrm{diff}}}

\subsection{Parameterized Difference Set and The multi-scale discrete Laplace Mechanism}

Recall that, for privacy analysis, we usually only consider the sensitivity of the function $q$, which is defined as $\Delta(q) = \max_{x,x'} |q(x)-q(x')|$ where the maximum is over all pairs $x$ and $x'$ differing on one entry. In this section, we show that, if we parameterized the potential different $q(x) - q(x')$ values in a more fine-grained manner, we can achieve an improved error in certain cases. 

For a given query function $q: X^d \to \mathbb{R}$, we define the \emph{difference set} of $q$, denoted by $\sdiff(q)$, as the set of all possible values of $|q(x) - q(x')|$ for all pairs $x$ and $x'$ differing on one entry. If $\sdiff(q)$ is finite, we let the $\sdiff(q)$-multi-scale discrete Laplace Mechanism to be the mechanism that outputs $q(x) + \sum_{i \in \sdiff(q)} i \cdot X_i$ where $X_i \overset{\text{i.i.d.}}{\sim} \DLap(\eps)$ for all $i \in \sdiff(q)$.

\begin{theorem} \label{thm:msdlap-diffset}
For any query function $q: X^d \to \mathbb{R}$ such that $\sdiff(q)$ is finite, the $\sdiff(q)$-multi-scale discrete Laplace Mechanism is $\eps$-DP. Furthermore, for $\eps \geq 1$, its MSE is $O\left(e^{-\eps} \cdot \sum_{i \in \sdiff(q)} i^2\right)$.
\end{theorem}

\begin{proof}
\emph{(Privacy)} To show that this mechanism is $\eps$-DP, it suffices to show that $\dr{q(x) + \sum_{i \in \sdiff(q)} i \cdot X_i}{q(x') + \sum_{i \in \sdiff(q)} i \cdot X_i} \leq \eps$ for any pair $x, x' \in X^d$ that differs on one entry. Consider any such fixed pair of $x, x'$. Let $\xi = q(x) - q(x')$; due to symmetry of the noise around zero, we may assume that $\xi \geq 0$. If $\xi = 0$, the statement is clearly true. Otherwise, from definition of $\sdiff$, we have $\xi \in \sdiff$.

From \Cref{lem:post-processing}, we have
\begin{align*}
&\dr{q(x) + \sum_{i \in \sdiff(q)} i \cdot X_i}{q(x') + \sum_{i \in \sdiff(q)} i \cdot X_i} \\
&=
\dr{\xi + \sum_{i\in \sdiff(q)} i \cdot X_i}{\sum_{i\in \sdiff(q)} i \cdot X_i} \\
&\leq \dr{\xi + \xi \cdot X_{\xi}}{\xi \cdot X_{\xi}} \\
&= \dr{1 + X_{\xi}}{X_{\xi}} \leq \eps,
\end{align*}
where the last inequality follows from $X_{\xi} \sim \DLap(\eps)$. Thus, the mechanism is $\eps$-DP as desired.

\emph{(Accuracy)} The MSE of the mechanism is
\begin{align*}
\Var\left(\sum_{i \in \sdiff(q)} i \cdot X_i\right) 
%&= \sum_{i=1}^{\Delta} \Var(i \cdot X_i)\\
&=  \sum_{i \in \sdiff(q)} i^2 \cdot \Var(X_i)\\ 
&= O\left(e^{-\eps} \cdot \sum_{i \in \sdiff(q)} i^2\right). & \qedhere
\end{align*}
\end{proof}

To see the advantage of the above mechanism, we note a few scenarios where
this mechanism has MSE $O(\Delta^2 \cdot e^{-\eps})$, but where approaches
which consider the sensitivity alone must have MSE at least $\Omega(\min\{\Delta^3 e^{-\eps}, \Delta^2 e^{-2\eps / 3}\})$ \cite{geng2014optimal, stair2015, stair2016adding}, which is asymptotically larger. In each example
we consider a query $q$ with sensitivity $\Delta$.

\begin{itemize}
\item $\sdiff(q) \subseteq [\lceil \Delta^{2/3} \rceil] \cup \{\Delta\}$, i.e. there is one large possible difference, but possibly many small ones far from $\Delta$.
\item $\sdiff(q) = \{n^m : m \le \log_n(\Delta) \in \mathbb{Z}_{\ge 0}\}$ for fixed $n, m > 0$ i.e. differences are structured to form exponential buckets.
\end{itemize}

Finally, we note that the setting where $|\sdiff(q)|$ is small (even when $\Delta(q)$ is large) can occur in practice. As an example, imagine a simple merchant that sells items whose prices are all in the set $S = \{5, 10, 30, 100\}$ and they want to privately sum a database of sales where each row is the sale price of a single item. If a neighboring dataset adds or removes a row, it is clear the sensitivity of this query is $\Delta=100$. For $\epsilon = 10$, the continuous staircase will have MSE\footnote{MSE results rounded to two significant figures.} 8.5, but the MSDLap mechanism with $\sdiff(q) = S$ will have MSE 1.0, resulting in nearly an order of magnitude improvement.

\end{document}